\documentclass[twocolumn]{IEEEtran}
\usepackage{amsmath}
\usepackage{graphicx}
\usepackage{latexsym}
\usepackage{amssymb}
\usepackage{cite}
\usepackage{color}
\usepackage{multirow}
\usepackage{algorithmic,algorithm}
\usepackage{footnote}
\usepackage{flushend}

    \def\Complex{{\rm\rule[.23ex]{.03em}{1.1ex}\kern-.3em{C}}}

    \newcommand{\be}{\begin{equation}} \newcommand{\ee}{\end{equation}}
    \newcommand{\bea}{\begin{eqnarray}} \newcommand{\eea}{\end{eqnarray}}
    \newcommand{\benum}{\begin{enumerate}} \newcommand{\eenum}{\end{enumerate}}

  \newtheorem{theorem}{Theorem}
\begin{document}

\title{Truth-Telling Mechanism for Secure Two-Way Relay Communications with Energy-Harvesting Revenue\thanks{This work is supported by EPSRC under grant EP/K015893/1.}}

\author{Muhammad R. A. Khandaker,~\IEEEmembership{Member,~IEEE}, Kai-Kit Wong,~\IEEEmembership{Fellow,~IEEE},

\thanks{M. R. A. Khandaker and K. K. Wong are with the Department of Electronic and Electrical Engineering, University College London, WC1E 7JE, United Kingdom (e-mail: $\rm m.khandaker@ucl.ac.uk$; $\rm kai\text{-}kit.wong@ucl.ac.uk$).} and Gan Zheng, ~\IEEEmembership{Senior Member,~IEEE},\thanks{G. Zheng is with The Wolfson School of Mechanical, Electrical and Manufacturing Engineering, Loughborough University, United Kingdom (e-mail: $\rm g.zheng@lboro.ac.uk$).}}

\maketitle \thispagestyle{empty}
\vspace*{-3em}

\begin{abstract}
This paper brings the novel idea of paying the utility to the winning agents in terms of some physical entity in cooperative communications. Our setting is a secret two-way communication channel where two transmitters exchange information in the presence of an eavesdropper. The relays are selected from a set of interested parties such that the secrecy sum rate is maximized. In return, the selected relay nodes' energy harvesting requirements will be fulfilled up to a certain threshold through their own payoff so that they have the natural incentive to be selected and  involved in the communication. However, relays may exaggerate their private information in order to improve their chance to be selected. Our objective is to develop a mechanism for relay selection that enforces them to reveal the truth since otherwise they may be penalized. We also propose a joint cooperative relay beamforming and transmit power optimization scheme based on an alternating optimization approach. Note that the problem is highly non-convex since the objective function appears as a product of three correlated Rayleigh quotients. While a common practice in the existing literature is to optimize the relay beamforming vector for given transmit power via rank relaxation, we propose a second-order cone programming (SOCP)-based approach in this paper which requires a significantly lower computational task. The performance of the incentive control mechanism and the optimization algorithm has been evaluated through numerical simulations.

\begin{center}
{\bf Index Terms}
\end{center}
Cooperative beamforming; energy harvesting; mechanism design; secrecy; two-way relay.
\end{abstract}



\section*{\sc I. Introduction}
Relaying is a promising technique to extend wireless coverage and increase the achievable rate \cite{jrnl_mur1, jrnl_su_para, tway_sec_dist, tway_sec_jam}, and in recent years it has also been recognized as a spectrally efficient way to exchange information over distance between two transceivers via two-way relaying \cite{tway_rate_reg, tway_spce_eff, jrnl_2way, tway_sec}. Relays, if used collaboratively, can also form focused signal or noise beams to provide physical-layer security \cite{secrecy_coop, tway_sec_dist, tway_sec_jam}. Collaborative relays follow the same idea as multiple antennas to exploit the spatial degrees of freedom for enhancing the signals to the legitimate receiver and worsening the interception of the eavesdropper by transmitting artificial noises \cite{goel_an, khisti_misome, qli_spatial, tway_sec_mimo}.


There is a huge scope of research for selecting the best relay nodes in maximizing the system performance. A meaningful setting would be to let the selected relays earn some form of revenue for relaying others' information. In this case, challenge arises because the candidates may behave selfishly to maximize their own revenues. To tackle this, game theory is a popular tool to analyze the conflict of interests among intelligent rational competitors \cite{gth_comm, gth_tway_jamm, gth_int_jamm}. Auction and pricing schemes were proposed for efficient selection of a social choice, but most of them were based on the assumption that the players are honest and ready to disclose their true private information \cite{gth_tway_jamm, gth_int_jamm}, which may not be the case in practice. Also, in the literature, the ``revenues'' are usually some abstract quantities that may not be meaningful \cite{gth_tway_jamm, gth_int_jamm, mechD_maga, mechd_caching, mechD_truth}.

Nevertheless, a recent development in wireless communications, which promotes energy transfer over wireless channels, may be the answer to help quantify the revenues one may gain from contributing to others' communications. Through simultaneous wireless information and power transfer (SWIPT), mobile users are provided with access to both energy and data at the same time which brings enormous prospects of new applications \cite{swipt_1st, swipt_s2t, swipt_bc, jrnl_swipt, jrnl_secrecy, jrnl_secrecy_sinr}. The concept of SWIPT was first introduced in \cite{swipt_1st} in a single noisy line, and later extended in \cite{swipt_s2t} to frequency-selective channels. Practical SWIPT schemes, namely, time switching and power splitting, have also been proposed \cite{swipt_bc, jrnl_swipt}. Recent studies further considered the combination of SWIPT with physical-layer security \cite{jrnl_secrecy, jrnl_secrecy_sinr}, one-way relaying \cite{swipt_nasir}, and two-way relaying \cite{tway_swipt}.

The focus of this paper is fundamentally different from the literature. While we consider relay selection for a two-way communication system in which two nodes exchange information with the help of a set of relay nodes in the presence of an eavesdropper, rather than concentrating primarily on reaping the benefits of relaying for secrecy communications, our aim is to develop an efficient mechanism to ensure that the relays reveal their true private information for relay selection optimization. In this particular problem, the channel coefficients from a relay to the two sources and the eavesdropper are regarded as the private information of that relay. The participation of relays is incentivised by the possible energy earning from the sources. In particular, the source transmitters will ensure that the energy harvesting requirements of the selected relays are fulfilled up to a certain threshold (or the expected payoff level).

The problem is that under this setup, the relays may exaggerate their private information to improve their chance to be selected, hoping to maximize their energy earning. The objective for a self-enforcing truth-revealing mechanism is to ensure that the relays reveal their actual private information to avoid being punished to pay for any damage caused. Note that mechanism design approaches have already been considered for suppressing cheating in cognitive radio networks \cite{mechD_maga}, wireless video caching \cite{ mechd_caching}, and one-way relaying \cite{mechD_truth}. However, in \cite{mechD_maga, mechd_caching, mechD_truth}, the revenue was paid in terms of some virtual entity, which does not directly relate to the concerned participants, while in this paper, the revenue is physically defined as {\em harvested energy}. In the context of energy harvesting facility considered in this paper, it is assumed that only the {\it selected} relays can harvest their {\it required} energy, and the {\it unselected} relay nodes will harvest almost {\it nothing}. It is also assumed that the relays will participate in the mechanism, as is common in conventional relaying \cite{tway_sec_dist, tway_sec_jam, tway_rate_reg, tway_spce_eff}, even in the absence of dedicated energy transmission. However, there is no guard mechanism to prevent any relay from announcing its undermined channel condition in an attempt \textit{not} to be selected so it can harvest energy without paying any penalty. In that case, the relay may remain unselected even with a better channel condition. But the reality is that the channel state information (CSI) of each relay is its own private information and none of the relays actually knows the channel conditions of the other relays. Hence none of them can define any threshold downplaying by which may guarantee its non-selection. Although it may be generally assumed that any unselected relay will be able to harvest some extent of energy, there is no guarantee that the harvested energy would be above a useful level. Thus the key motivation for the relays to participate in the mechanism is that through the proposed mechanism they yield QoS guarantee (at least minimum incentive) in terms of energy earning. On the other hand, the unselected relays have no such guarantee.

With the mechanism, we then propose a joint collaborative relay beamforming and transmit power optimization scheme for maximizing the sum secrecy rate while guaranteeing the expected payoff of each selected relay node in the form of its harvested energy. The optimization problem appears to be highly non-convex as the objective function is a product of three correlated Rayleigh quotients. While a common practice tends to optimize the collaborative relay beamforming vector for a given transmit power using rank relaxation, our proposed approach requires no rank relaxation. Instead, we formulate the relay beamforming problem as a second-order cone program (SOCP), which has lower computational overhead.

To the best of our knowledge, the closest work in the existing literature to this paper can be found in \cite{mechD_truth}. However, our contribution is three-fold compared to the work in \cite{mechD_truth}. Firstly, we consider two-way amplify-and-forward relaying, whereas one-way decode-and-forward (DF) relaying was considered in \cite{mechD_truth}. The DF relaying vastly simplifies the utility characterization for mechanism design. Hence the system model is different. Secondly, we define the utility of the auctioneers (relays) in terms of some practically appealing quantity (harvested energy) as opposed to the virtual payment considered in \cite{mechD_truth} and many other existing works \cite{ mechd_caching}. Note that the virtual payment system does not provide enough incentives to the players for participating in the auction. Thirdly, in addition to the incentive controlling mechanism design, we develop an optimal joint transmit power and relay beamforming design algorithm whereas \cite{mechD_truth} considered only truthful mechanism design for relay selection. We also note that collaborative relay beamforming problems for two-way relay systems were studied in \cite{tway_sec_dist, tway_sec_jam} but with a fixed number of relays, and without mechanism design and payments for the selected relays in terms of harvested energy.

The remainder of this paper is organized as follows. In Section~II, the system model for a two-way relay network in the presence of an eavesdropper is described. Truth-telling mechanism design strategies are then briefly introduced in Section~III. The joint-optimal collaborative relay beamforming and transmit power optimization algorithm is developed in Section~IV. Section~V presents the simulation results to illustrate the importance of the proposed mechanism design and we conclude the paper in Section~VI.

{\em Notations}---Throughout the paper, boldface lowercase and uppercase letters are used to represent vectors and matrices, respectively. The symbol ${\bf I}_n$ denotes an $n\times n$ identity matrix, while $\bf 0$ is a zero vector or matrix. Also, ${\bf A}^T$, ${\bf A}^H$, ${\bf A}^\dag$, ${\rm tr}({\bf A})$, ${\rm rank}({\bf A})$, and ${\rm det}({\bf A})$ represent transpose, the Hermitian (conjugate) transpose, matrix projection, trace, rank and determinant of a matrix ${\bf A}$, respectively; $\|\cdot\|$ represents the Euclidean norm; ${\bf A}\succeq {\bf 0}\, ({\bf A}\succ {\bf 0})$ means that ${\bf A}$ is a Hermitian positive semidefinite (definite) matrix; $[{\bf A}]_{i,j}$ denotes the $(i,j)$th element of ${\bf A}$. The notation ${\bf x}\sim \mathcal{CN}(\boldsymbol{\mu}, {\boldsymbol\Sigma})$ means that ${\bf x}$ is a random vector following a complex circularly symmetric Gaussian distribution with the mean vector $\boldsymbol{\mu}$ and the covariance matrix of ${\boldsymbol\Sigma}$.

\begin{figure}[ht]
\centering
\includegraphics*[width=\columnwidth]{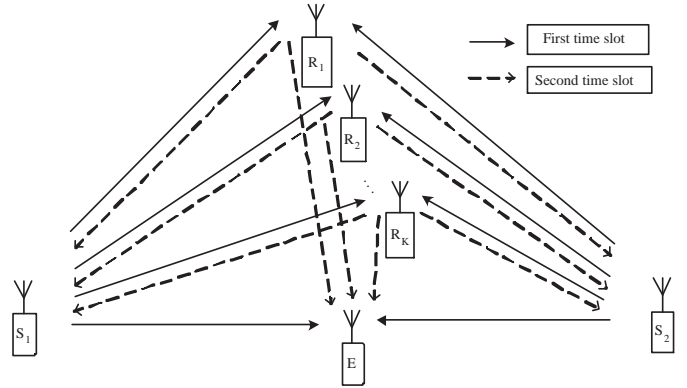}
\caption{A two-way relay system in the presence of an eavesdropper.}\label{sysmod}
\end{figure}

\section*{\sc II. System Model}\label{sec_sys}
We consider a two-way relay network consisting of two sources, ${\sf S}_1$ and ${\sf S}_2$, wishing to communicate with each other, $N$ relay nodes, $\{{\sf R}_i\}_{i=1}^N$, and an eavesdropper, ${\sf E}$, as illustrated in Fig.~\ref{sysmod}. There is no direct link between the two source nodes, so communication has to be done via the relays. Assuming the more practical half-duplex relays, the communication is accomplished in two time slots. In the first time slot, the source nodes broadcast their signals $s_1$ and $s_2$ to all the relay nodes. In the second time slot, the source nodes decide which of those $N$ relays will be selected to forward their messages to the corresponding destination nodes based on some predesigned mechanism which we will describe later. During the whole process, the eavesdropper node overhears the messages from the source nodes as well as the relay nodes. The source nodes aim at maximizing the secrecy sum-rate by properly selecting $K\leq N$ relay nodes. It is assumed that each relay node only knows its own CSI between itself and the transmitters as well as the eavesdropper. The relays then report their CSI to the mechanism designer (which may be one of the two sources or a centralized processor)\footnote{Note that the same node performs the transmit power and relay beamforming optimization and/or relay selection operations as well.} as their bids to be selected.

The messages, $s_1$ and $s_2$, transmitted from the sources need to be kept confidential to ${\sf E}$. It is assumed that $s_1$ and $s_2\sim\mathcal{CN}(0,1)$, and the transmit power from ${\sf S}_1$ and ${\sf S}_2$ is, respectively, $p_{{\rm s},1}$ and $p_{{\rm s},2}$. In the first time slot, the received signals at ${\sf R}_i$ and ${\sf E}$ are, respectively, given by
\begin{align}
y_{{\rm r},i}& = \sqrt{p_{{\rm s},1}}h_{1,i}s_1 + \sqrt{p_{{\rm s},2}}h_{2,i}s_2 + n_{{\rm r},i}, ~\mbox{for }i=1,\dots, N, \label{yri}\\
y_{{\rm e}}^{(1)}& = \sqrt{p_{{\rm s},1}}h_{1,{\rm e}}s_1 + \sqrt{p_{{\rm s},2}}h_{2,{\rm e}}s_2 + n_{{\rm e}}^{(1)},\label{ye1}
\end{align}
where $h_{i,j}$ for $i = 1, 2$ and $j = 1, \dots, N,$ denote the complex channel gains between ${\sf S}_i$ and ${\sf R}_j$ and $h_{i,{\rm e}}$ for $i = 1, 2$, are that between ${\sf S}_i$ and ${\sf E}$, $n_{{\rm r},i}\sim\mathcal{CN}(0,\sigma^2)$ and $n_{{\rm e}}^{(1)}\sim\mathcal{CN}(0,\sigma^2)$ represent the complex additive white Gaussian noises (AWGNs) at ${\sf R}_i$ and ${\sf E}$ during the first time slot, respectively.

In vector form, the signals received at all the relays can be expressed as
\begin{equation}
{\bf y}_{{\rm r}} = \sqrt{p_{{\rm s},1}}{\bf h}_{1,{\rm r}}s_1 + \sqrt{p_{{\rm s},2}}{\bf h}_{2,{\rm r}}s_2 + {\bf n}_{{\rm r}},
\end{equation}
where ${\bf h}_{1,{\rm r}} \triangleq \left[h_{1,1}, \dots, h_{1,N}\right]^T$, ${\bf h}_{2,{\rm r}} \triangleq \left[h_{2,1}, \dots, h_{2,N}\right]^T$ denote the channel vectors between the two sources and the relays, and ${\bf n}_{{\rm r}}\triangleq \left[n_{{\rm r},1}, \dots, n_{{\rm r},N}\right]^T$ indicates the AWGN vector at the relay nodes. We assume that each relay node is equipped with a power splitting device to coordinate harvesting energy and forwarding the received signal. In particular, the received signal at the $i$th relay, ${\sf R}_i$, is split such that a $\rho_i\in[0,1]$ portion of the signal power is passed to the information forwarding block and the remaining $1-\rho_i$ portion of the power is sent to the energy harvesting block of the relay. Several power splitting schemes have been considered in the literature \cite{swipt_bc, jrnl_swipt} including fixed power splitting and dynamic power splitting. In order to keep our main focus on mechanism design, we consider fixed power splitting in this paper. Interested readers are referred to \cite{swipt_bc, jrnl_swipt} for more about the dynamic power splitting schemes.

From \eqref{yri}, the harvested power at the $i$th relay node, ${\sf R}_i$, is given by
\begin{equation}
P_{{\rm h},i} = \xi_i(1-\rho_i)\left(p_{{\rm s},1}|h_{1,i}|^2 + p_{{\rm s},2}|h_{2,i}|^2 + \sigma^2\right),
\end{equation}
where $\xi_i\in(0,1]$ denotes the energy conversion efficiency of the energy transducers at the $i$th relay that accounts for the loss in the energy transducers for converting the harvested energy to electrical energy to be stored. For convenience, we assume, without loss of generality, that $\xi_k = 1, \forall k$, in this paper. It is worth pointing out that the relays do not need to convert the received signal from the radio frequency (RF) band to the baseband in order to harvest the carried energy using modern energy transducers. Therefore, according to the law of energy conservation, it is assumed that the total harvested RF band power (energy normalized by the baseband symbol period) at each relay is proportional to the normalised energy of the received baseband signal.

In the second time slot, ${\sf R}_i$ amplifies the received signal $\sqrt{\rho_i}y_{{\rm r},i}$ by a complex weighting coefficient $f^*_i$ and then transmits $x_{{\rm r},i} = \sqrt{\rho_i}f^*_iy_{{\rm r},i}$. Combining the transmit signals from all the relay nodes, we have ${\bf x}_{{\rm r}} = {\bf F}{\bf y}_{{\rm r}}$ where ${\bf F}$ is the combined diagonal weight matrix in the form ${\bf F} = {\rm diag}\left({\bf f}^*\right)$, with ${\bf f}\triangleq \left[\sqrt{\rho_1}f_1, \dots, \sqrt{\rho_N}f_N\right]^T$. Note that for notational simplicity, the power splitting coefficients have been incorporated in the definition of the relay beamforming vector ${\bf f}$.
It is also assumed that the channel coefficients between the transmitters and the relays are block-fading reciprocal. The block-fading reciprocal channel assumption has been widely used in two-way relay literature, e.g., \cite{tway_sec_dist, tway_sec_jam, tway_rate_reg}. The assumption essentially means that channels for the two phases are reciprocal, which is based on the time division duplex (TDD) operation with synchronized time-slot. The TDD operation greatly reduces signalling overhead and leads to an SOCP-based problem formulation with reduced complexity, which we will elaborate in section~IV. Thus, the received signal at ${\sf S}_1$ in the second time slot can be expressed as
\begin{multline}
y_{{\rm s},1} = {\bf h}_{1,{\rm r}}^T{\bf x}_{\rm r} + n_{{\rm s},1} = \sqrt{p_{{\rm s},1}}{\bf h}_{1,{\rm r}}^T{\bf F}{\bf h}_{1,{\rm r}}s_1\\ 
+ \sqrt{p_{{\rm s},2}}{\bf h}_{1,{\rm r}}^T{\bf F}{\bf h}_{2,{\rm r}}s_2 + {\bf h}_{1,{\rm r}}^T{\bf F}{\bf n}_{{\rm r}} + n_{{\rm s},1},
\end{multline}
where $n_{{\rm s},1}\sim\mathcal{CN}(0,\sigma^2)$ denotes the AWGN signal at source node ${\sf S}_1$.

Similarly, the received signal at ${\sf S}_2$ can be expressed as
\begin{multline}
y_{{\rm s},2} = {\bf h}_{2,{\rm r}}^T{\bf x}_{\rm r} + n_{{\rm s},2} = \sqrt{p_{{\rm s},1}}{\bf h}_{2,{\rm r}}^T{\bf F}{\bf h}_{1,{\rm r}}s_1\\ 
+ \sqrt{p_{{\rm s},2}}{\bf h}_{2,{\rm r}}^T{\bf F}{\bf h}_{2,{\rm r}}s_2 + {\bf h}_{2,{\rm r}}^T{\bf F}{\bf n}_{{\rm r}} + n_{{\rm s},2},
\end{multline}
and that at ${\sf E}$ can be written as
\begin{multline}
y_{\rm e}^{(2)} = {\bf h}_{\rm r, e}^T{\bf x}_{\rm r} + n_{\rm e}^{(2)} = \sqrt{p_{{\rm s},1}}{\bf h}_{\rm r, e}^T{\bf F}{\bf h}_{1,{\rm r}}s_1\\ 
+ \sqrt{p_{{\rm s},2}}{\bf h}_{\rm r, e}^T{\bf F}{\bf h}_{2,{\rm r}}s_2 + {\bf h}_{\rm r, e}^T{\bf F}{\bf n}_{{\rm r}} + n_{\rm e}^{(2)},\label{ye2}
\end{multline}
where $n_{{\rm s},2}\sim\mathcal{CN}(0,\sigma^2)$ and $n_{\rm e}^{(2)}\sim\mathcal{CN}(0,\sigma^2)$ are the noises at ${\sf S}_2$ and ${\sf E}$ in the second time slot.

Since $s_1$ and $s_2$ are known, respectively, at ${\sf S}_1$ and ${\sf S}_2$, the residual received signals after self-interference cancellation (typical for two-way channels) are, respectively, given by
\begin{align}
y_{{\rm s},1} &= \sqrt{p_{{\rm s},2}}{\bf h}_{1,{\rm r}}^T{\bf F}{\bf h}_{2,{\rm r}}s_2 + {\bf h}_{1,{\rm r}}^T{\bf F}{\bf n}_{{\rm r}} + n_{{\rm s},1}\\ 
&= \sqrt{p_{{\rm s},2}}{\bf f}^H{\bf H}_{1,{\rm r}}{\bf h}_{2,{\rm r}}s_2 + {\bar n}_{{\rm s},1}\\ 
&= \sqrt{p_{{\rm s},2}}{\bf f}^H{\bf h}_{2,1}s_2 + {\bar n}_{{\rm s},1},
\end{align}
and
\begin{align}
y_{{\rm s},2} &= \sqrt{p_{{\rm s},1}}{\bf h}_{2,{\rm r}}^T{\bf F}{\bf h}_{1,{\rm r}}s_1 + {\bf h}_{2,{\rm r}}^T{\bf F}{\bf n}_{{\rm r}} + n_{{\rm s},2}\\ 
&= \sqrt{p_{{\rm s},1}}{\bf f}^H{\bf H}_{2,{\rm r}}{\bf h}_{1,{\rm r}}s_1 + {\bar n}_{{\rm s},2}\\ 
&= \sqrt{p_{{\rm s},1}}{\bf f}^H{\bf h}_{1,2}s_1 + {\bar n}_{{\rm s},2},
\end{align}
where ${\bf H}_{i,{\rm r}} \triangleq {\rm diag}({\bf h}_{i,{\rm r}})$, ${\bf h}_{j,i} \triangleq {\bf H}_{i,{\rm r}}{\bf h}_{j,{\rm r}}$, for $i,j = 1, 2$, and $j \neq i$, ${\bar n}_{{\rm s},i} \triangleq {\bf h}_{i,{\rm r}}^T{\bf F}{\bf n}_{{\rm r}} + n_{{\rm s},i}$, for $i = 1, 2$, and we have used the identity ${\bf a}^H{\rm diag}({\bf b}) = {\bf b}^H{\rm diag}({\bf a})$. Note that each transmission phase brings some opportunity for ${\sf E}$ to overhear the information. Hence, combining the received signals in \eqref{ye1} and \eqref{ye2} at ${\sf E}$ over two time slots, an equivalent multiple-input multiple-output (MIMO) channel is formed, i.e.,
\begin{multline}
\begin{array}{c}\underbrace{\left[\begin{array}{c}
y_{\rm e}^{(1)}\\
y_{\rm e}^{(2)}
\end{array}
\right]}=\\
{\bf y}_{\rm e}
\end {array}
\begin{array}{c}\underbrace{\left[\begin{array}{cc}
\sqrt{p_{{\rm s},1}}h_{1,{\rm e}} & \sqrt{p_{{\rm s},2}}h_{2,{\rm e}}\\
\sqrt{p_{{\rm s},1}}{\bf f}^H{\bar{\bf h}}_{1,{\rm e}} & \sqrt{p_{{\rm s},2}}{\bf f}^H{\bar{\bf h}}_{2,{\rm e}}
\end{array}
\right]}\\
{\bf H}_{\rm e}
\end {array} \begin{array}{c}\underbrace{\left[\begin{array}{c}
s_1\\
s_2
\end{array}
\right]}\\
{\bf s}
\end {array}
\\
\begin{array}{c}\underbrace{+\left[\begin{array}{c}
{n}_{\rm e}^{(1)}\\
{\bar n}_{\rm e}^{(2)}
\end{array}
\right]},\\
{\bf n}_{\rm e}
\end {array}\label{ye}
\end{multline}
where ${\bar{\bf h}}_{i,{\rm e}} \triangleq {\bf H}_{\rm r, e}{\bf h}_{i,{\rm r}}$, for $i = 1, 2$, ${\bf H}_{\rm r, e} \triangleq {\rm diag}({\bf h}_{\rm r, e})$, and ${\bar n}_{\rm e}^{(2)} \triangleq {\bf h}_{\rm r, e}^T{\bf F}{\bf n}_{{\rm r}} + n_{\rm e}^{(2)}$.

As a result, the corresponding signal-to-noise ratio (SNR) for the equivalent transmission link from ${\sf S}_2$ to ${\sf S}_1$ can be expressed as
\begin{equation}
\gamma_{1} = \frac{p_{{\rm s},2}{\bf f}^H{\bf h}_{2,1}{\bf h}_{2,1}^H{\bf f}}{\sigma^2\left({\bf f}^H{\bf C}_{{\rm n},1}{\bf f} + 1\right)},
\end{equation}
where ${\bf C}_{{\rm n},1} \triangleq {\bf H}_{1,{\rm r}}{\bf H}_{1,{\rm r}}^H$. Similarly, the SNR for the equivalent transmission link from ${\sf S}_1$ to ${\sf S}_2$ is
\begin{equation}
\gamma_{2} = \frac{p_{{\rm s},1}{\bf f}^H{\bf h}_{1,2}{\bf h}_{1,2}^H{\bf f}}{\sigma^2\left({\bf f}^H{\bf C}_{{\rm n},2}{\bf f} + 1\right)}
\end{equation}
with ${\bf C}_{{\rm n},2} \triangleq {\bf H}_{2,{\rm r}}{\bf H}_{2,{\rm r}}^H$. Thus, the channel capacities at ${\sf S}_1$, ${\sf S}_2$, and ${\sf E}$ are given, respectively, by
\begin{align}
C_{1} &= \frac{1}{2}\log_2\left(1 + \gamma_{1}\right),\\
C_{2} &= \frac{1}{2}\log_2\left(1 + \gamma_{2}\right),
\end{align}
and
\begin{equation}\label{ce}
C_{{\rm e}} = \frac{1}{2}\log_2\det\left({\bf I}_2 + {\bf H}_{\rm e}{\bf H}_{\rm e}^H{\bf C}_{\rm n,e}^{-1}\right),
\end{equation}
where ${\bf C}_{\rm n,e} \triangleq {\rm diag}\left(\sigma^2,\sigma^2\left(1 + {\bf f}^H{\bf H}_{\rm r,e}{\bf f}\right)\right)$ is the equivalent noise covariance matrix at the eavesdropper ${\sf E}$ over the two time slots and the scalar factor $\frac{1}{2}$ is due to the fact that two time slots are required in order to accomplish one successful transmission. Then the achievable secrecy sum rate is given by \cite{tway_sec_dist, tway_sec_jam}
\begin{equation}
C_{{\rm s}} = \left[C_{1} + C_{2} - C_{{\rm e}}\right]^+ \label{sec_sum}
\end{equation}
where $[a]^+ = \max(0,a)$. Note that the secrecy sum-rate in \eqref{sec_sum} is the sum of secrecy rates provided by all the relay nodes. Since all the relay nodes may not have sufficiently strong fading channels in order to make a useful contribution to the secrecy sum-rate, selecting the appropriate relays as helpers can play a significant role in improving secrecy performance. In the next section, we will focus on the mechanism design approach in order to select the $K$ best relays that can make the most significant contribution.

However, since the relays selected will have greater opportunity\footnote{Note that in the proposed beamforming algorithm, the transmitters will transmit with sufficient power such that the energy harvesting requirements of all the selected relay nodes are satisfied at least to equality assuming that the relays report their true channel information.} to harvest energy from the received signal, all the relays will be naturally interested in participating in the mechanism. The issue is that some of them may intentionally exaggerate their true information in order to be selected. We will focus on the incentive control mechanisms so that the participating relays are self-enforced to reveal the truth.

\section*{\sc III. Truth-Telling Mechanism Design}\label{sec_mechD}
This section provides a brief introduction of mechanism design. A mechanism $\mathcal M$ is defined by the tuple $\left({\mathcal S}, t_1, \dots, t_N\right)$ where $t_i$ for $i = 1, \dots, N,$ represents the transfer payment of agent $i$ (or player $i$)\footnote{In this paper, the terms ``player'' and ``agent'' will be used interchangeably.} when the social choice is ${\mathcal S}$. The transfer payment is the compensation paid by an agent in return to the social damage it causes to the others by being selected. Mechanism design (sometimes called reverse game theory) is a game theoretical tool that studies solutions for a class of private information games in order to achieve a specific system-wide outcome even though the agents are selfish \cite{nobel07}. In a mechanism, each agent reports its private information (referred to as `type' in the native literature) to the designer that serves as the parameter of a valuation function quantifying its bid on a specific allocation outcome and the transfer payment. The most desirable criteria that the mechanism designers tend to achieve are incentive compatibility and social optimality. A mechanism is said to be incentive compatible if truth-telling becomes the dominant (best) strategy in the mechanism while the mechanism is social optimum if it can ensure the maximum aggregate utilities of all the agents in the system. The Vickrey-Clarke-Groves (VCG) mechanism \cite{vickery, clarke, groves} is well known to achieve these two goals. Hence, we consider the VCG mechanism in the relay selection problem in order to maximize the secrecy sum-rate.

\subsection*{A. VCG Mechanism}\label{subsec_vcg}
In the VCG mechanism, agents are the members of the society. All the agents announce their valuations for the auctioned items simultaneously. Hence, there is no way to know whether the agents are telling the truth. The design objective is to give the agents the right incentives to tell the truth. The social choice is a set of $K$ agents from a set of $N$ alternatives for $K$ identical auctioned items. In VCG mechanism, each winning agent must pay some compensation (i.e., transfer payment) for the social damage it causes. The more the damage, the higher is the transfer payoff. We will now present the framework to quantify how much each agent $i$ contributes to the rest of the society if selected.

Let $v_i\left({\mathcal X}, \theta_i\right)$ denote the valuation by agent $i$ from alternative ${\mathcal X}$ given the true information $\theta_i$. We also denote ${\mathcal O}(\hat\theta_i,\hat{{\theta}}_{-i})$ as the utilitarian alternative (i.e., outcome of the mechanism) chosen from the available set of alternatives based on the reported information $\{\hat\theta_i\}_{i = 1}^N$, as opposed to the true information $\{\theta_i\}_{i = 1}^N$, where the variable $\hat{{\theta}}_{-i} \triangleq \{\hat\theta_1, \dots, \hat\theta_{i-1}, \hat\theta_{i+1}, \dots, \hat\theta_{N}\}$ is defined as the set of reported information of all the agents except agent $i$. Also, ${\mathcal O}_{-i}(\hat\theta_j,\hat{{\theta}}_{-j})$ represents the utilitarian alternative when agent $i$ does not take part in the mechanism. Note that the type profile $\hat{{\theta}} \triangleq \{\hat\theta_1, \dots, \hat\theta_{N}\}$ is an ordered list in the decreasing manner.

The total welfare of the society (excluding $i$) is thus given by $\sum_{j\neq i}^Kv_j\left({\mathcal O}(\hat\theta_j,\hat{{\theta}}_{-j}),\theta_j\right)$. If agent $i$ were not a member of the society, then the social welfare would be changed to $\sum_{j=1}^Kv_j\left({\mathcal O}_{-i}(\hat\theta_j,\hat{{\theta}}_{-j}),\theta_j\right)$. The difference in the social welfare with and without the presence of agent $i$ is a measure of how much agent $i$ contributes to the rest of the society. In the VCG mechanism, agent $i$ receives a monetary transfer payment equal to the amount it contributes to the rest of the society. As a result, the VCG mechanism is characterized by the following monetary transfer payment function
\begin{align}
t_i(\hat\theta_i,\hat{{\theta}}_{-i}) &= \sum_{j\neq i}^Kv_j\left({\mathcal O}(\hat\theta_j,\hat{{\theta}}_{-j}),\theta_j\right)\nonumber\\
&- \sum_{j=1}^Kv_j\left({\mathcal O}_{-i}(\hat\theta_j,\hat{{\theta}}_{-j}),\theta_j\right)\label{trans_vcg}\\
&= \sum_{j = 1}^Kv_j\left({\mathcal O}(\hat\theta_j,\hat{{\theta}}_{-j}),\theta_j\right)\nonumber\\ 
&- \sum_{j=1}^Kv_j\left({\mathcal O}_{-i}(\hat\theta_j,\hat{{\theta}}_{-j}),\theta_j\right)\nonumber\\
&- v_i\left({\mathcal O}(\hat\theta_i,\hat{{\theta}}_{-i}),\theta_i\right).\label{trans_vcg1}
\end{align}
Note that the two summation operations in \eqref{trans_vcg} and \eqref{trans_vcg1} are conducted within two different sets of alternatives namely ${\mathcal O}(\hat\theta_i,\hat{{\theta}}_{-i})$ and ${\mathcal O}_{{-i}}(\hat\theta_j,\hat{{\theta}}_{-j})$. The first sum $\sum_{j\neq i}^Kv_j\left({\mathcal O}(\hat\theta_j,\hat{{\theta}}_{-j}),\theta_j\right)$ in (17) includes ($K-1$) terms while the second sum $\sum_{j = 1}^Kv_j\left({\mathcal O}_{-i}(\hat\theta_j,\hat{{\theta}}_{-j}),\theta_j\right)$ includes $K$ different terms. Thus given a type profile $\hat\theta$, the monetary transfer to agent $i$ is defined by the total value of all agents other than $i$ when agent $i$ is present in the system minus the total value of all agents when agent $i$ is absent in the system. The value is always negative since the sum of apparently (in absence of the $i$th item) highest $K$ valuations is subtracted from the sum of the highest $(K-1)$ valuations. Note that the transfer payment of agent $i$ is independent of its own valuation $v_i$. The difference of the first two terms in \eqref{trans_vcg1} represents the marginal contribution of agent $i$ to the system which is given as a discount to agent $i$ by the VCG payment mechanism. It is evident from \eqref{trans_vcg1} that all the $K$ winning bidders pay a social damage recovery payment equal to the highest non-winning (i.e., the $(K+1)$-st) bid, whereas a losing bidder pays nothing, i.e.,
\begin{equation}
t_i(\hat\theta_i,\hat{{\theta}}_{-i}) = \left\{\begin{array}{l}\!-v_{K+1}\left({\mathcal O}(\hat\theta_j,\hat{{\theta}}_{-j}),\theta_j\right), ~\text{for  } k = 1, \dots, K,\\
\! 0, \qquad\qquad\text{for  } k = K+1, \dots, N. \end{array}\right.
\end{equation}

In the VCG mechanism, the highest $K$ bidders win and the winning bidder $i$ attains a utility (payoff) of
\begin{eqnarray}
u_i\left(\hat\theta_i,\hat\theta_{-i}\right) \!\!\!&=&\!\!\! v_i\left({\mathcal O}\left(\hat\theta_i,\hat\theta_{-i}\right),\theta_i\right) + t_i\left(\hat\theta_i,\hat\theta_{-i}\right) \label{util_vcg}\\
\!\!\!&=&\!\!\! \sum_{i = 1}^Kv_i\left({\mathcal O}\left(\hat\theta_i,\hat\theta_{-i}\right),\theta_i\right)\nonumber\\
\!\!\!& &\!\!\! - \sum_{j=1}^Kv_j\left({\mathcal O}_{-i}\left(\hat\theta_j,\hat\theta_{-j}\right),\theta_j\right).\label{util_vcg1}
\end{eqnarray}
Note that the penalty method to prevent reporting false information by agent $i$ is imposed by the transfer payment $t_i\left(\hat\theta_i,\hat\theta_{-i}\right)$ in \eqref{util_vcg} which distinguishes mechanism design from conventional game theory. In conventional game theory, the agents can exaggerate their private information arbitrarily in order to be selected such that their own payoff is maximized. But in the VCG mechanism, the transfer payment will penalize them if they do so. Thus the selected utilitarian alternative maximizes the sum of the announced valuations, i.e., $$\sum_{i = 1}^Kv_i\left({\mathcal O}\left(\hat\theta_i,\hat\theta_{-i}\right),\theta_i\right) \geq \sum_{j = 1}^Kv_j\left({\mathcal O}_{-i}\left(\hat\theta_j,\hat\theta_{-j}\right),\theta_j\right),$$
where the equality holds only when all the agents reveal their true private information.
Let us now elaborate the VCG payment mechanism through a simple numerical example.

\subsubsection*{Example~1 VCG Transfer Payment}
Consider five agents $\{1, 2, 3, 4, 5\}$ with valuations $v_1 = 22$, $v_2 = 18$, $v_3 = 15$, $v_4 = 12$ and $v_5 = 8$ participating in a sealed bid auction for three identical items available for auction. Each bidder can bid for one item only. Applying the VCG mechanism, bidders $1, 2,$ and $3$ should win since their bids confirm the maximum social welfare $(22 + 18 + 15 = 55)$. The transfer payment by bidder $1$ is calculated as
\begin{eqnarray}
t_1\left(\hat\theta_1,\hat\theta_{-1}\right) \!\!\!&=&\!\!\! \sum_{j\neq 1}^3v_j\left({\mathcal O}\left(\hat\theta_j,\hat\theta_{-j}\right),\theta_j\right)\nonumber\\
\!\!\!& &\!\!\! - \sum_{j = 1}^3v_j\left({\mathcal O}_{-1}\left(\hat\theta_j,\hat\theta_{-j}\right),\theta_j\right)\nonumber\\
\!\!\!&=&\!\!\! (18 + 15) - (18 + 15 + 12)\nonumber\\
\!\!\!&=&\!\!\! -12.\nonumber
\end{eqnarray}
Thus bidder $1$ pays an amount $(12)$ equal to the highest non-winning bid $v_4 = 12$ for the social damage caused by its selection. Similarly, the transfer payments paid by bidders $2$ and $3$ both equal to $12$. Note that the payments are consistent with their respective marginal contributions. The marginal contribution of agent $1$ is given by
\begin{eqnarray}
\sum_{j = 1}^3 \!\!\!&&\!\!\! \!\!\!\!\!\!v_j\left({\mathcal O}\left(\hat\theta_j,\hat\theta_{-j}\right),\theta_j\right) - \sum_{j = 1}^3v_j\left({\mathcal O}_{-1}\left(\hat\theta_j,\hat\theta_{-j}\right),\theta_j\right)\nonumber\\
\!\!\!&=&\!\!\! (22 + 18 + 15) - (18 + 15 + 12)\nonumber\\
\!\!\!&=&\!\!\! 10 \nonumber
\end{eqnarray}
which is given as a discount to agent $1$ resulting in a transfer payment of $10 - 22 = -12$. Similarly, the marginal contribution of agents $2$ and $3$ can be computed as $(22 + 18 + 15) - (22 + 15 + 12) = 6$ and $(22 + 18 + 15) - (22 + 18 + 12) = 3$.

Thus the utilities of the agents can be computed as $u_1 = 22 - 12 = 10, u_2 = 18-12 = 6, u_3 = 15-12 = 3, u_4 = 0, u_5 = 0$.

Let us now assume that agent $4$ announces an exaggerated valuation of $v_4 = 22$, as opposed to its true valuation $12$, with a desire to win. Thus the agents $\{1, 2, 4\}$ win and their transfer payments can be obtained as $t_1 = t_2 = t_4 = -15$, which is equal to the highest non-winning bid. The corresponding payoffs of the winning bids are computed as $u_1 = 22-15 = 7, u_2 = 18 - 15 = 3, u_4 = 12 - 15 = -3$. Note that a negative utility of agent $4$ indicates that the agent must pay additional amount from its own pocket in order to comply with the auction rules. Now the total social welfare counts to $\sum_{i=1}^5 u_i = 7 + 3 -3 + 0 + 0 = 7$ as opposed to $19$ if all the agents would have announced their true valuations. Thus the VCG mechanism gives the incentives that if any of the agents announces untrue valuation, that may damage the total social benefit as well as its own utility.\hfill $\blacksquare$

In the following, we apply the VCG mechanism for relay selection in a two-way communication system in presence of an eavesdropper.

\subsection*{B. VCG Mechanism for Relay Selection}
We consider the channel coefficients of each relay node with the two source nodes and the eavesdropping node as the private information of that relay node. The relay nodes report their channel information $\hat{g}_i \triangleq \{\hat{h}_{1,i}, \hat{h}_{2,i}, \hat{h}_{{\rm e},i}\}$ to the source nodes (or the mechanism designer) simultaneously. Through reporting their CSI, the relay nodes actually commit to the mechanism designer the level of secrecy rates they can provide for the two source nodes. We assume that the selected relay nodes must keep their commitments during their transmission in the second phase. Although the reported information may not be the same as the true ones, the mechanism designer will select the relays treating them as true. Let $g_i \triangleq \{{h}_{1,i}, {h}_{2,i}, {h}_{{\rm e},i}\}$ denote ${\sf R}_i$'s true channel information and $C_{i,{\rm s}}(g_i)$ denote the achievable secrecy sum rate through relay ${\sf R}_i$. Note that the information leakage during the first time slot is not affected by the social choice of relays and we assume that the relays do cooperative null space beamforming towards the eavesdropper's channel.\footnote{This will be elaborated in Section~IV} Hence, $C_{i,{\rm s}}(g_i)$ can be defined as a function of the equivalent two-way single-input single-output (SISO) channel only. After removing the self-interference, the equivalent SISO channel from ${\sf S}_2$ to ${\sf S}_1$ via ${\sf R}_i$ can be modelled as
\begin{equation}
{\tilde y}_{{\rm s},1} = \alpha_i\sqrt{p_{\rm r}p_{{\rm s},2}}h_{1,i}h_{2,i}s_2 + \alpha_i\sqrt{p_{\rm r}}h_{1,i}n_{{\rm r},i} + n_{{\rm s},1}
\end{equation}
and that from ${\sf S}_1$ to ${\sf S}_2$ is given by
\begin{equation}
{\tilde y}_{{\rm s},2} = \alpha_i\sqrt{p_{\rm r}p_{{\rm s},1}}h_{2,i}h_{1,i}s_1 + \alpha_i\sqrt{p_{\rm r}}h_{2,i}n_{{\rm r},i} + n_{{\rm s},2},
\end{equation}
where $\alpha_i \triangleq \left(p_{{\rm s},1}|h_{1,i}|^2 + p_{{\rm s},2}|h_{2,i}|^2 + \sigma^2\right)^{-\frac{1}{2}}$ is the amplification factor satisfying the power constraint at relay $i$ and $p_{\rm r}$ is the available relay power budget. Thus ${\sf R}_i$'s independent valuation can be defined as
\begin{multline}
v_i\left(g_i\right) \triangleq C_{i,{\rm s}}\left(g_i\right) = \frac{1}{2}\left[\log_2\left(1 + \frac{\alpha_i^2p_{\rm r}p_{{\rm s},1}|h_{1,i}|^2|h_{2,i}|^2}{\sigma^2\left(\alpha_i^2p_{\rm r}|h_{1,i}|^2 + 1\right)}\right)\right.\\\left. + \log_2\left(1 + \frac{\alpha_i^2p_{\rm r}p_{{\rm s},2}|h_{2,i}|^2|h_{1,i}|^2}{\sigma^2\left(\alpha_i^2p_{\rm r}|h_{2,i}|^2 + 1\right)}\right)\right]. \label{csi}
\end{multline}
Note that by dividing the numerator and the denominator of both logarithmic terms in \eqref{csi} by $\alpha_i^2p_{\rm r}$, $C_{i,{\rm s}}\left(g_i\right)$ can be shown as an increasing function of $p_{{\rm s},1}, p_{{\rm s},2}$ and $p_{\rm r}$. Hence during the mechanism design phase, we obtain $C_{i,{\rm s}}(g_i)$ assuming $p_{{\rm s},1}, p_{{\rm s},2}$ and $p_{\rm r}$ hold their maximum possible value. Thus the utilitarian alternative ${\mathcal O}\left(\hat{g}_i, \hat{g}_{-i}\right)$ based on the reported channel information can be defined as
\begin{equation}
{\mathcal O}\left(\hat{g}_i, \hat{g}_{-i}\right) \triangleq \arg ~~ \max_{\{{\sf R}_k\}} \sum_{k=1}^KC_{i,{\rm s}}(\hat{g}_i). \label{selec_cri}
\end{equation}
Note that based on the definitions of the two sets ${\mathcal O}(\cdot)$ and ${\mathcal O}_{-i}(\cdot)$, the output in  \eqref{selec_cri} of the proposed mechanism design is a set $\{{\sf R}_k\}$ of $K$ relay nodes.

Let us define that $\pi_i$ is the average harvested power (price paid) against per unit of secrecy rate achieved by relay $i$. It is worth mentioning that the unit price $\pi_i$ may vary amongst the relays depending on their channel fading conditions. Thus the utility of ${\sf R}_i$ can be defined independently as
\begin{equation}
u_i\left(\hat{g}_i\right) = \left\{\begin{array}{l}\pi_i C_{i,{\rm s}}\left(\hat{g}_i\right), \qquad \text{if  } {\sf R}_i \text{  is selected},\\
0, \qquad\qquad\qquad \text{otherwise.} \end{array}\right.
\end{equation}

Note that in the existing game-theoretic approaches adopted in secrecy communication, the agents receive some virtual payment usually in terms of secrecy rate or transmit power \cite{mechD_truth}, which has no operational meaning to them. However, we propose the utility to be paid through some physical entity (e.g., harvested energy) for the first time. In this paper, we assume that only the relay nodes selected by the mechanism designer can get payoff i.e., harvest required energy from the first time slot. Although this may not always be the case in practice, it is a valid (reasonable) assumption since the mechanism designer selects the relays with the best channel conditions. Essentially, the unselected relay nodes, which have worse channel conditions as guaranteed by the proposed mechanism design, will harvest almost nothing. Applying energy beamforming\footnote{We do not consider energy beamforming in this paper. Readers are referred to \cite{tway_swipt, jrnl_secrecy} for energy beamforming strategies.} \cite{jrnl_secrecy} at both transmitting nodes, one can fully guarantee that the {\it unselected} relays will not be able to harvest any energy from the transmitters' signals. However, designing such spatially selective energy beamforming is a complicated task \cite{qli_spatial, tway_swipt, jrnl_secrecy} and requires additional resources (e.g., physical antennas) at the two transmitters, which is not compatible with the system settings (single-antenna transmitters) considered in this paper. Hence, in order to keep the main focus of this paper on mechanism design, we would like to leave transmit energy beamforming design as a potential future work. Since only the selected relay nodes can get payoff, some dishonest relays may exaggerate their channel information in order to create greater opportunity to be selected. This may result in an unfair selection and damage the expected payoff of the unselected relay nodes. Essentially, this will adversely affect the secrecy sum rate and no equilibrium can be achieved under this condition \cite{clarke, mechD_truth}. Hence we aim at designing a useful mechanism that can assist in controlling the incentives of the relays through imposing some penalty functions for the dishonest relay nodes. The penalty function will ensure that if any relay node is selected based on its exaggerated channel information, it will pay more transfer payment for the social damage caused from its own source of power in order to guarantee the required level of secrecy rate at each source node.

In order to better clarify the motivation that drives the relays to exaggerate their true valuations (i.e., CSI in this case), we introduce the probability of being selected affecting their valuation decision. The higher the valuation, the higher the probability of being selected, and so is the expected payoff. In this context, we assume that the relay nodes do not know the channel information of the other relays before they actually enact their channel information but generally know that the secrecy rate of each relay obeys certain probability density function $\left(0\leq C_{i,{\rm s}}(g_i)<\infty\right)$. Thus we define the reported valuation of ${\sf R}_i$ as
\begin{equation}
v_i\left({\mathcal O}(\hat g_i,\hat{g}_{-i}),g_i\right) \triangleq C_{i,{\rm s}}\left(\hat{g}_i\right){\rm Pr}\left({\sf R}_i\text{ being selected}\right),
\end{equation}
where ${\rm Pr}(A)$ indicates the probability that the event $A$ occurs. Accordingly, the expected payoff of ${\sf R}_i$ can be defined as
\begin{equation}
\tilde u_i\left(\hat{g}_i\right) \triangleq \pi_i C_{i,{\rm s}}\left(\hat{g}_i\right){\rm Pr}\left({\sf R}_i\text{ being selected}\right).
\end{equation}
Given the relay selection criterion \eqref{selec_cri}, the natural incentive of a relay would thus be to exaggerate its achievable secrecy rate $C_{i,{\rm s}}\left(\hat{g}_i\right)$ to $\infty$ in order to get the maximum expected payoff, which eventually increases their probability of being selected. Hence we introduce the following VCG transfer payment function
\begin{multline}\label{trans_fun}
t_i\left(\hat g_i,\hat g_{-i}\right)=\sum_{j\neq i}^Kv_j\left({\mathcal O}\left(\hat g_j,\hat g_{-j}\right),g_j\right)\\ - \sum_{j = 1}^Kv_j\left({\mathcal O}_{-i}\left(\hat g_j,\hat g_{-j}\right),g_j\right),
\end{multline}
where ${\mathcal O}_{-i}(\cdot)$ is the relay selection outcome when ${\sf R}_i$ does not participate in the mechanism. It is obvious from \eqref{trans_fun} that if a relay node claims a higher secrecy rate by tempering $\hat{h}_{1,i}$ or $\hat{h}_{2,i}$, it may have more chances to be selected, but runs the risk of paying extra transfer payoff through spending from its own source of power.\footnote{The exact mechanism to implement this will be discussed in Section~IV.} On the other hand, if a relay node reports a lower secrecy rate, it will receive a higher monetary compensation but at the cost of lower probability to be selected. Hence truth-telling is the dominant strategy in the proposed VCG mechanism. The idea will be elaborated in Section~V through numerical examples. In the VCG mechanism based relay selection algorithm, the total payoff of ${\sf R}_i$ is given by
\begin{multline}
u_i\left(\hat g_i,\hat g_{-i}\right) = v_i\left({\mathcal O}\left(\hat g_i,\hat g_{-i}\right),g_i\right) + t_i\left(\hat g_i,\hat g_{-i}\right)\\ = \sum_{j = 1}^Kv_j\left({\mathcal O}\left(\hat g_j,\hat g_{-j}\right),g_j\right)\\ - \sum_{j = 1}^Kv_j\left({\mathcal O}_{-i}\left(\hat g_j,\hat g_{-j}\right),g_j\right).\label{util_rel1}
\end{multline}
The following theorem describes the strength of the VCG mechanism for relay selection.

\begin{theorem}
Announcing truthfully, i.e., $\hat{g}_i = g_i$ is a dominant strategy for each relay $i$.
\end{theorem}

\begin{proof}
We need to prove that announcing $\hat{g}_i = g_i$ is the best strategy for relay $i$ no matter what other relays announce. If relay ${\sf R}_i$ announces $\hat{g}_i$ and others announce $\hat g_{-i}$, then according to \eqref{util_rel1}, ${\sf R}_i$'s utility is $u_i\left(\hat g_i,\hat g_{-i}\right) = v_i\left({\mathcal O}(\hat g_i,\hat{g}_{-i}),g_i\right) + \sum_{j \ne i}^Kv_j\left({\mathcal O}(\hat g_j,\hat{g}_{-j}),g_j\right) - \sum_{j = 1}^Kv_j\left({\mathcal O}_{-i}(\hat g_j,\hat{g}_{-j}),g_j\right)$. Relay $i$ has to decide which $\hat{g}_i$ to announce; however, it cannot determine ${\mathcal O}_{-i}(\hat g_j,\hat{g}_{-j})$ since it is excluded from that society. Hence, we can ignore the last term in $u_i\left(\hat g_i,\hat g_{-i}\right)$ as it is unaffected by ${\sf R}_i$'s announcement. Therefore, in order to maximize its own payoff, relay ${\sf R}_i$ aims to maximize the total utility of the society inclusive of itself. Since relay ${\sf R}_i$ cannot choose other relays' announcements, it can only play its own part. That is, by truthfully announcing, $\hat{g}_i = g_i$, it can ensure that ${\mathcal O}( g_i, \hat{g}_{-i})$ will be chosen. Hence announcing truthfully is the best thing relay ${\sf R}_i$ can do.
\end{proof}

Note that each relay node competing to be selected will have the same incentive to report its true CSI and the $K$ relays that can achieve the top $K$ secrecy rates will be selected which will eventually maximize the total payoff. Thus equilibrium is achieved under this condition.

Interestingly, the only additional task for implementing the proposed mechanism in relay selection, as opposed to conventional relay selection, is the calculation of the transfer payments, which involves simple mathematical operations. In return, the benefit is that the mechanism enforces the relays to reveal their true CSI. No additional signalling is needed since the node performing the optimization and/or relay selection can effectively implement the mechanism. A quantitative comparison of benefits has been provided in {\em Example~1}. As demonstrated in the example, if agent $4$ announces an exaggerated valuation, the total social welfare counts to $7$ as opposed to $19$ if all the agents would have announced their true valuations. Thus the VCG mechanism gives the incentives that if any of the agents announces untrue valuation, that may damage the total social benefit as well as its own utility.

Once the best relays are selected based on their reported channel information, the  optimization of the transmit power and cooperative relay beamforming is conducted, which we discuss in the next section.

\section*{\sc IV. Optimal Transmit Power and Relay Beamforming Design}\label{sec_opt_alg}
In this section, we propose transmit power and cooperative relay beamforming optimization schemes assuming that full CSI of all the nodes is available. Although in some practical communication systems, obtaining the eavesdropper's CSI can be very difficult (or even impossible), for the ease of exposition, we assume that the relays know their channels with the transmitters as well as the eavesdropper. This is a reasonable assumption for scenarios where the eavesdropper is an active user of the system, and the transmitter aims to provide different services to different types of users. For such active eavesdroppers, the CSI can be estimated from the eavesdropper's transmission. Let us define $P_{{\rm b},i} \triangleq P_{{\rm h},i} - |f_i|^2\left(p_{{\rm s},1}|h_{1,i}|^2 + p_{{\rm s},2}|h_{2,i}|^2 + \sigma^2\right)$ as the net power to be stored in the battery of the $i$th relay. The overall objective is to increase $C_1$ and $C_2$ as well as $P_{{\rm b},i}$ as much as possible while keeping $C_{\rm e}$ as small as possible under peak power constraints at the two transmitters as well as each relay node. Hence we formulate the following optimization problem
\begin{subequations}\label{rateP1}
\begin{eqnarray}
\max_{p_{{\rm s},1},p_{{\rm s},2},{\bf f}} \!\!\!& &\!\!\! \left[C_{1} + C_{2} - C_{{\rm e}}\right]^+ + \min_i P_{{\rm b},i}\label{rateP1_o}\\
{\rm s.t.} \!\!\!& &\!\!\! P_{{\rm h},i} \geq u_i\left(\hat g_i,\hat g_{-i}\right),~\mbox{for }i = 1, \dots, K,\label{rateP1_c1}\\
\!\!\!& &\!\!\! |f_i|^2\left(p_{{\rm s},1}|h_{1,i}|^2 + p_{{\rm s},2}|h_{2,i}|^2 + \sigma^2\right)\leq p_{\rm r},\nonumber\\
\!\!\!& &\!\!\! \qquad\qquad\qquad\qquad~\mbox{for }i = 1, \dots, K,\label{rateP1_c2}\\
\!\!\!& &\!\!\! p_{{\rm s},1} \leq P_{\rm max},~~p_{{\rm s},2} \leq P_{\rm max}. \label{rateP1_c3}
\end{eqnarray}
\end{subequations}
Here $P_{\rm max}$ and $p_{\rm r}$ are the available power budgets at the two sources and each of the relay nodes, respectively. Note that the last term in \eqref{rateP1_o} indicates the saved power of the worst selected relay. In general, it may happen that $P_{{\rm b},i}$ is negative, which essentially means that the $i$th selected relay may need to contribute additional power from its own storage in order to maintain its reported secrecy rate. However, the constraint \eqref{rateP1_c1} ensures that each of the selected relays gets its appropriate payoff. To guarantee that the relay nodes do not need to use their own source of power, they may set $p_{\rm r}\leq u_i\left(\hat g_i,\hat g_{-i}\right)$. Then the constraints \eqref{rateP1_c1} and \eqref{rateP1_c2} jointly guarantee that the honest selected relays can harvest sufficient energy required for their transmission in the second phase. However, there is no guarantee that a dishonest relay will be able to harvest appropriate amount of energy since they likely have weaker fading channels than what they have reported. Since we assume that the selected relays transmit with sufficient power during the second phase such that their promised secrecy rates at two sources are maintained, only the honest relay nodes do not need to utilize their own source of power. Although $u_i\left(\hat g_i,\hat g_{-i}\right)$ can assume any value in a general sense, we obtain $u_i\left(\hat g_i,\hat g_{-i}\right)$ from \eqref{util_rel1} assuming $p_{{\rm s},1} = p_{{\rm s},2} = P_{\rm max}$.

Note that the objective function in \eqref{rateP1_o} includes the product of three correlated Rayleigh quotients, which is neither convex, nor concave, and is in general very difficult to solve. However, a more tractable but suboptimal strategy for designing beamforming is to choose the beamforming vector lying in the null space of the eavesdropper's channel in the second time slot. The corresponding beamforming optimization problem is to maximize the sum rate achieved at two sources instead of sum secrecy rate. Because we cannot cancel the information rate leakage to the eavesdropper during the first time slot, the impact of the eavesdropper's achievable information rate on the secrecy sum rate should be considered when optimizing the beamforming vector as well as two source powers. As such, we can try to degrade the eavesdropper's interception by constraining its maximum allowable information rate with a predetermined level $r_{\rm e}$, which can help avoid dealing with the rate difference of concave functions in \eqref{rateP1_o}.
If the relay nodes choose the beamforming vector ${\bf f}$ lying in the null space of the eavesdropper's equivalent channel vectors, then the information leackage in the second phase is completely eliminated, i.e., ${\bf f}^H{\bar{\bf h}}_{1,{\rm e}} = {\bf f}^H{\bar{\bf h}}_{2,{\rm e}} = 0$ so that the second row of ${\bf H}_{\rm e}$ in \eqref{ye} can be eliminated.
Thus $C_{\rm e}$ reduces to
\begin{eqnarray}
C_{\rm e} = \frac{1}{2}\log_2\left(1 + \frac{p_{{\rm s},1}|h_{1,{\rm e}}|^2 + p_{{\rm s},2}|h_{2,{\rm e}}|^2}{\sigma^2}\right).
\end{eqnarray}
Introducing a real-valued slack variable $\nu$, we reformulate problem \eqref{rateP1} as
\begin{subequations}\label{rateP2}
\begin{eqnarray}
\max_{p_{{\rm s},1},p_{{\rm s},2},{\bf f}, \nu} \!\!\!& &\!\!\! \frac{1}{2}\log_2\left(1 + \frac{p_{{\rm s},2}{\bf f}^H{\bf h}_{2,1}{\bf h}_{2,1}^H{\bf f}}{\sigma^2\left(1 + {\bf f}^H{\bf C}_{{\rm n},1}{\bf f}\right)}\right)\nonumber\\
\!\!\!& &\!\!\! + \frac{1}{2}\log_2\left(1 + \frac{p_{{\rm s},1}{\bf f}^H{\bf h}_{1,2}{\bf h}_{1,2}^H{\bf f}}{\sigma^2\left(1 + {\bf f}^H{\bf C}_{{\rm n},2}{\bf f}\right)}\right) + \nu \label{rateP2_o}\\
{\rm s.t.} \!\!\!& &\!\!\! \frac{1}{2}\log_2\!\left(\!1 \!+ \frac{p_{{\rm s},1}|h_{1,{\rm e}}|^2 + p_{{\rm s},2}|h_{2,{\rm e}}|^2}{\sigma^2}\right)\! \leq \! r_{\rm e}\label{rateP2_c1}\\
\!\!\!& &\!\!\! \left(1 - \rho_i\right)\left(p_{{\rm s},1}|h_{1,i}|^2 + p_{{\rm s},2}|h_{2,i}|^2 + \sigma^2\right) \geq u_i\nonumber\\
\!\!\!& &\!\!\! \qquad\qquad\times\left(\hat g_i,\hat g_{-i}\right),~\mbox{for }i = 1, \dots, K,\label{rateP2_c2}\\
\!\!\!& &\!\!\! |f_i|^2\left(p_{{\rm s},1}|h_{1,i}|^2 + p_{{\rm s},2}|h_{2,i}|^2 + \sigma^2\right)\leq p_{\rm r},\nonumber\\
\!\!\!& &\!\!\! \qquad\qquad\qquad\qquad~\mbox{for }i = 1, \dots, K,\label{rateP2_c3}\\
\!\!\!& &\!\!\! P_{{\rm b},i} \ge \nu,~\mbox{for }i = 1, \dots, K,\label{rateP2_c4}\\
\!\!\!& &\!\!\! p_{{\rm s},1} \leq P_{\rm max},~~ p_{{\rm s},2} \leq P_{\rm max}, \label{rateP2_c5}
\end{eqnarray}
\end{subequations}
where ${\bf f} = {\bar{\bf H}}_{\rm e}^\dag{\bar{\bf f}}$, ${\bar{\bf f}}$ is any vector, ${\bar{\bf H}}_{\rm e}^\dag$ is the projection matrix onto the null space of ${\bar{\bf H}}_{\rm e} \triangleq \left[{\bar{\bf h}}_{1,{\rm e}}, {\bar{\bf h}}_{2,{\rm e}}\right]$, the columns of which constitute the orthogonal basis for the null space of ${\bar{\bf H}}_{\rm e}$. Note that from \eqref{rateP2_c3}, the transmit power of the $i$th relay node can be expressed as $\left[{\bf f}{\bf f}^H\right]_{i,i}\left[{\bf R}_{\rm s}\right]_{i,i}$ with ${\bf R}_{\rm s} = p_{{\rm s},1}{\bf H}_{1,{\rm r}}{\bf H}_{1,{\rm r}}^H + p_{{\rm s},2}{\bf H}_{2,{\rm r}}{\bf H}_{2,{\rm r}}^H + \sigma^2{\bf I}_K$.
Also, for given $p_{{\rm s},1}$ and $p_{{\rm s},2}$, we can see from \eqref{rateP2} that \eqref{rateP2_c1}, \eqref{rateP2_c2}, and \eqref{rateP2_c5} are irrelevant to ${\bf f}$. However, the problem is still non-convex since the objective function is not concave. Hence we split the objective function and formulate the following relay beamforming optimization problem
\begin{subequations}\label{rateP3}
\begin{eqnarray}
\max_{{\bf f}, r_0, \nu} \!\!\!& &\!\!\! r_0 + \nu \label{rateP3_o}\\
{\rm s.t.} \!\!\!& &\!\!\! \frac{1}{2}\log_2\left(1 + \frac{p_{{\rm s},2}{\bf f}^H{\bf h}_{2,1}{\bf h}_{2,1}^H{\bf f}}{\sigma^2\left(1 + {\bf f}^H{\bf C}_{{\rm n},1}{\bf f}\right)}\right) \geq \beta r_0\label{rateP3_c1}\\
\!\!\!& &\!\!\! \frac{1}{2}\log_2\!\left(1 + \frac{p_{{\rm s},1}{\bf f}^H{\bf h}_{1,2}{\bf h}_{1,2}^H{\bf f}}{\sigma^2\left(1 + {\bf f}^H{\bf C}_{{\rm n},2}{\bf f}\right)}\right)\! \geq \!(1\!-\!\beta)r_0\label{rateP3_c2}\\
\!\!\!& &\!\!\! \left[{\bf f}{\bf f}^H\right]_{i,i}\left[{\bf R}_{\rm s}\right]_{i,i} \leq p_{\rm r},~\mbox{for }i = 1, \dots, K,\label{rateP3_c3}\\
\!\!\!& &\!\!\! \left(1 - \rho_i - \left[{\bf f}{\bf f}^H\right]_{i,i}\right)\left[{\bf R}_{\rm s}\right]_{i,i} \ge \nu,\nonumber\\
\!\!\!& &\!\!\! \qquad\qquad\qquad\qquad~\mbox{for }i = 1, \dots, K,\label{rateP3_c4}
\end{eqnarray}
\end{subequations}
where $r_0$ is the objective value for the sum rate in \eqref{rateP2_o} and $\beta \in [0,1]$ is the rate splitting coefficient. The optimal solution of the problem can be found in two steps. First we solve problem \eqref{rateP3} for a feasible $r_0$ to obtain ${\bf f}$. Then we perform a one-dimensional search on $\beta$ to find the maximum $r_0$ for which problem \eqref{rateP3} is feasible. The lower bound of the rate search is definitely $0$. However, to define the upper bound $r_{\rm max}$, we first decouple the two-way relay channel into two one-way relay channels and obtain the rate $r_i$ of each one-way channel. Then the upper limit can be defined as $r_{\rm max} = 2\times {\rm max}(r_1, r_2)$. Let us now substitute ${\bf f} = {\bar{\bf H}}_{\rm e}^\dag{\bar{\bf f}}$ in \eqref{rateP3} to obtain
\begin{subequations}\label{rateP4}
\begin{eqnarray}
\max_{{\bar{\bf f}}, r_0, \nu} \!\!\!& &\!\!\! r_0 + \nu \label{rateP4_o}\\
{\rm s.t.} \!\!\!& &\!\!\! \frac{{\bar{\bf f}}^H{\bar{\bf H}}_{\rm e}^{\dag H}{\bf h}_{2,1}{\bf h}_{2,1}^H{\bar{\bf H}}_{\rm e}^\dag{\bar{\bf f}}}{1 + {\bar{\bf f}}^H{\bar{\bf H}}_{\rm e}^{\dag H}{\bf C}_{{\rm n},1}{\bar{\bf H}}_{\rm e}^\dag{\bar{\bf f}}} \geq \frac{\sigma^2}{p_{{\rm s},2}}\left(2^{2\beta r_0}-1\right),\label{rateP4_c1}\\
\!\!\!& &\!\!\! \frac{{\bar{\bf f}}^H{\bar{\bf H}}_{\rm e}^{\dag H}{\bf h}_{1,2}{\bf h}_{1,2}^H{\bar{\bf H}}_{\rm e}^\dag{\bar{\bf f}}}{1 + {\bar{\bf f}}^H{\bar{\bf H}}_{\rm e}^{\dag H}{\bf C}_{{\rm n},2}{\bar{\bf H}}_{\rm e}^\dag{\bar{\bf f}}} \geq \frac{\sigma^2}{p_{{\rm s},1}}\!\left(\!2^{2(1 - \beta) r_0}\!-\!1\!\right),\label{rateP4_c2}\\
\!\!\!& &\!\!\! \left[{\bar{\bf H}}_{\rm e}^\dag{\bar{\bf f}}{\bar{\bf f}}^H{\bar{\bf H}}_{\rm e}^{\dag H}\right]_{i,i}\!\left[{\bf R}_{\rm s}\right]_{i,i} \!\leq\! p_{\rm r},\mbox{for }i = 1, \dots, K,\label{rateP4_c3}\\
\!\!\!& &\!\!\! \left(1 - \rho_i - \left[{\bar{\bf H}}_{\rm e}^\dag{\bar{\bf f}}{\bar{\bf f}}^H{\bar{\bf H}}_{\rm e}^{\dag H}\right]_{i,i}\right)\left[{\bf R}_{\rm s}\right]_{i,i} \ge \nu,\nonumber\\
\!\!\!& &\!\!\! \qquad\qquad\qquad\qquad\mbox{for }i = 1, \dots, K. \label{rateP4_c4}
\end{eqnarray}
\end{subequations}
Problem \eqref{rateP4} is a non-convex quadratically constrained quadratic programming (QCQP) problem which is $NP$-hard in general. We reformulate problem \eqref{rateP4} as follows:
\begin{subequations}\label{rateP5}
\begin{eqnarray}
\max_{{\bar{\bf f}}, r_0, \nu} \!\!\!& &\!\!\! r_0 + \nu \label{rateP5_o}\\
{\rm s.t.} \!\!\!& &\!\!\! \left|{\bar{\bf f}}^H{\bar{\bf H}}_{\rm e}^{\dag H}{\bf h}_{2,1}\right|^2 \geq \eta_1\left\|\left[\begin{array}{c}\sqrt{{\bf C}_{{\rm n},1}}{\bar{\bf H}}_{\rm e}^\dag{\bar{\bf f}}\\
1\end{array}\right]\right\|^2,\label{rateP5_c1}\\
\!\!\!& &\!\!\! \left|{\bar{\bf f}}^H{\bar{\bf H}}_{\rm e}^{\dag H}{\bf h}_{1,2}\right|^2 \geq \eta_2\left\|\left[\begin{array}{c}\sqrt{{\bf C}_{{\rm n},2}}{\bar{\bf H}}_{\rm e}^\dag{\bar{\bf f}}\\
1\end{array}\right]\right\|^2,\label{rateP5_c2}\\
\!\!\!& &\!\!\! \left|{\bar{\bf H}}_{\rm e}^{\dag(i)}{\bar{\bf f}}\right|^2 \leq \frac{p_{\rm r}}{\left[{\bf R}_{\rm s}\right]_{i,i}},~\mbox{for }i = 1, \dots, K,\label{rateP5_c3}\\
\!\!\!& &\!\!\! \left|{\bar{\bf H}}_{\rm e}^{\dag(i)}{\bar{\bf f}}\right|^2 \!\leq\! 1 \!-\! \rho_i \!-\!\frac{\nu}{\left[{\bf R}_{\rm s}\right]_{i,i}}, \mbox{for }i = 1, \dots, K,\label{rateP5_c4}
\end{eqnarray}
\end{subequations}
where $\eta_1 \triangleq \sigma^2\left(2^{2\beta r_0}-1\right)/p_{{\rm s},2}$, $\eta_2 \triangleq \sigma^2\left(2^{2(1 - \beta) r_0}-1\right)/p_{{\rm s},1}$, $\sqrt{{\bf C}_{{\rm n},i}}$ is the element-wise square root of ${\bf C}_{{\rm n},i}$, and ${\bar{\bf H}}_{\rm e}^{\dag(i)}$ indicates the $i$th row of ${\bar{\bf H}}_{\rm e}^{\dag}$. Since the constraints in \eqref{rateP5} are expressed in terms of Euclidean vector norms, multiplying the optimal ${\bar{\bf f}}$ by an arbitrary phase shift $e^{j\phi}$ will not affect the constraints. Also, by definition, ${\bf h}_{2,1}$ and ${\bf h}_{1,2}$ yield identical numeric value. Therefore, ${\bar{\bf f}}^H{\bar{\bf H}}_{\rm e}^{\dag H}{\bf h}_{i,j}$ can be considered as a real number, without loss of generality. Consequently, \eqref{rateP5} can be rewritten as
\begin{subequations}\label{rateP6}
\begin{eqnarray}
\max_{{\tilde{\bf f}}, r_0, \nu} \!\!\!& &\!\!\! r_0 + \nu \label{rateP6_o}\\
{\rm s.t.} \!\!\!& &\!\!\! \left\|\tilde{\bf C}_{{\rm n},1}\tilde{\bf f}\right\| \leq \frac{1}{\sqrt{\eta_1}}\tilde{\bf h}_{2,1}^H \tilde{\bf f}, \label{rateP6_c1}\\
\!\!\!& &\!\!\! \left\|\tilde{\bf C}_{{\rm n},2}\tilde{\bf f}\right\| \leq \frac{1}{\sqrt{\eta_2}}\tilde{\bf h}_{1,2}^H \tilde{\bf f}, \label{rateP6_c2}\\
\!\!\!& &\!\!\! \left|{\tilde{\bf h}}_{{\rm e},i}^H{\tilde{\bf f}}\right| \leq \sqrt{\frac{p_{\rm r}}{\left[{\bf R}_{\rm s}\right]_{i,i}}},~\mbox{for }i = 1, \dots, K,\label{rateP6_c3}\\
\!\!\!& &\!\!\! \left|{\tilde{\bf h}}_{{\rm e},i}^H{\tilde{\bf f}}\right| \leq\! \sqrt{1 \!-\! \rho_i\! -\! \frac{\nu}{\left[{\bf R}_{\rm s}\right]_{i,i}}}, \mbox{for }i = 1, \dots, K, \label{rateP6_c4}
\end{eqnarray}
\end{subequations}
where $\tilde{\bf f} \triangleq \left[{\bar{\bf f}}^T, 1\right]^T$, $\tilde{\bf h}_{i,j}^H = \left[\left({\bar{\bf H}}_{\rm e}^{\dag H}{\bf h}_{i,j}\right)^H, 0\right]$, ${\tilde{\bf h}}_{{\rm e},i} = \left[{\bar{\bf H}}_{\rm e}^{\dag(i)}, 0\right]^T$, and $\tilde{\bf C}_{{\rm n},i} = \left[\begin{array}{cc}\sqrt{{\bf C}_{{\rm n},i}}{\bar{\bf H}}_{\rm e}^{\dag H} & {\bf 0}\\
{\bf 0} & 1 \end{array} \right]$. Note that \eqref{rateP6} is a standard SOCP problem which can be efficiently solved by interior point methods \cite{boyd}. Once the optimal relay beamforming vector ${\bf f}$ is obtained, we formulate the following problem using the monotonic property of the $\log$ function to find the optimal $p_{{\rm s},1}$ and $p_{{\rm s},2}$:
\begin{subequations}\label{powerP1}
\begin{eqnarray}
\max_{p_{{\rm s},1},p_{{\rm s},2},\nu} \!\!\!& &\!\!\! \mu_1 p_{{\rm s},1} + \mu_2 p_{{\rm s},2} + \nu,\label{powerP1_o}\\
{\rm s.t.} \!\!\!& &\!\!\! p_{{\rm s},1}|h_{1,{\rm e}}|^2 + p_{{\rm s},2}|h_{2,{\rm e}}|^2 \leq \sigma^2\left(2^{2 r_{\rm e}}-1\right),\label{powerP1_c1}\\
\!\!\!& &\!\!\! \left(1-\rho_i\right)\left(p_{{\rm s},1}|h_{1,i}|^2 + p_{{\rm s},2}|h_{2,i}|^2 + \sigma^2\right)\nonumber\\
\!\!\!& &\!\!\! \qquad\qquad \geq u_i\left(\hat g_i,\hat g_{-i}\right),~\mbox{for }i = 1, \dots, K,\label{powerP1_c2}\\
\!\!\!& &\!\!\! p_{{\rm s},1}|h_{1,i}|^2 + p_{{\rm s},2}|h_{2,i}|^2 + \sigma^2 \le \frac{p_{\rm r}}{|f_i|^2},\nonumber\\
\!\!\!& &\!\!\! \qquad\qquad\qquad\qquad \mbox{for }i = 1, \dots, K,\label{powerP1_c3}\\
\!\!\!& &\!\!\! \left(1 - \rho_i - |f_i|^2\right)\left(p_{{\rm s},1}|h_{1,i}|^2 + p_{{\rm s},2}|h_{2,i}|^2 + \sigma^2\right)\nonumber\\
\!\!\!& &\!\!\! \qquad\qquad\qquad  \ge \nu,~\mbox{for }i = 1, \dots, K,\label{powerP1_c4}\\
\!\!\!& &\!\!\! p_{{\rm s},1} \leq P_{\rm max},~~p_{{\rm s},2} \leq P_{\rm max}, \label{powerP1_c5}
\end{eqnarray}
\end{subequations}
where $\mu_i = \frac{{\bf f}^H{\bf h}_{i,j}{\bf h}_{i,j}^H{\bf f}}{\sigma^2\left(1 + {\bf f}^H{\bf C}_{{\rm n},j}{\bf f}\right)}, i, j = 1, 2, i\neq j$.
The problem \eqref{powerP1} is convex for given $\rho_i$ and hence the globally optimal solution can be easily obtained using existing solvers \cite{cvx}. Thus we update the relay beamforming vector ${\bf f}$ and the transmit powers $p_{{\rm s},1}$ and $p_{{\rm s},2}$ alternatingly. Since we solve a convex subproblem at each step of the alternating algorithm, the objective function can either increase or maintain, but cannot decrease at each step of the algorithm. A monotonic convergence follows directly from this observation. The algorithm is summarized in Table~\ref{tab_alg}.

\begin{table}[!ht]
\centering \caption{Proposed alternating algorithm for solving problem \eqref{rateP1}} \label{tab_alg}
\begin{tabular}[t]{|c|l|}
\hline
Step & Action\\
\hline
1 & Initialize $p_{{\rm s},1} = p_{{\rm s},2} = p_{\rm r} = \frac{P_{\rm max}}{K+2}$.\\
\hline
2 & Repeat\\
\hline
& a) Solve the SOCP problem \eqref{rateP6} using existing solvers,\\
& \quad e.g., CVX \cite{cvx}.\\
& b) Solve the linear programming problem \eqref{powerP1}.
 \\
\hline
3 & Until \it{convergence}\\
\hline
\end{tabular}
\end{table}

\subsection{Complexity of the Algorithm}
We now focus on the computational complexity of the proposed optimization scheme. We analyze the complexity of the alternating algorithm step-by-step. Note that the relay beamforming optimization problem \eqref{rateP6} involves only SOC constraints, and hence can be solved using standard interior-point methods (IPM) \cite[Lecture 6]{nemiro_cvx_opt}. Therefore, we can use the worst-case computation time of IPM to analyze the complexity of the proposed method.  Now the overall complexity of the IPM for solving an SOCP problem containing $p$ constraints consists of two components:
\begin{itemize}
    \item[a)] \textit{Iteration Complexity}: The number of iterations required to reach an $\epsilon$-accurate ($\epsilon > 0$) optimal solution is in the order of $\ln(1/\epsilon)\sqrt{\beta(\mathcal{K})}$, where $\beta(\mathcal{K}) = 2p $ is known to be the barrier parameter.

    \item[b)] \textit{Per-Iteration Computation Cost}: A system of $n$ linear equations is required to be solved in each iteration where $n$ is the number of decision variables. The computation tasks include the formation of the coefficient matrix $\bf H$ of the system of linear equations and the factorization of $\bf H$. The cost of forming $\bf H$ sums on the order of $\kappa_{\rm for}=n \sum_{j=1}^p k_j^2$, $k_j$ is the dimension of the $j$th cone, while the cost of factorization is on the order of $\kappa_{\rm fac}=n^3$ \cite{nemiro_cvx_opt}.
\end{itemize}
Thus the overall computation cost for solving the problem using IPM is on the order of $\ln(1/\epsilon)\sqrt{\beta(\mathcal{K})}\\ \times(\kappa_{\rm for} + \kappa_{\rm fac})$. Using these concepts, we can now analyze the computational complexity of problem (37). Note that the number of decision variables $n$ is on the order of $K$ (ignoring the slack variables). Now, the problem (37) has $p = (2K + 2)$ SOC constraints. Thus the complexity of solving problem (37) is on the order of $4K\sqrt{(K+1)}\mathcal{O}(K)[(K+1)^2 + K^2 +1]\ln(1/\epsilon)$.

In the next step of the algorithm, problem \eqref{powerP1} is solve, which is a linear programming problem. Now the linear program \eqref{powerP1} can be solved in polynomial time at a worst-case complexity of $\mathcal{O}\left(3^{3.5}(3K+3)^2 \right)$ \cite{lin_prog_ipm}.

\section*{\sc V. Simulation Results}\label{sec_sim}
In this section, we study the performance of the proposed mechanism design and joint source-relay optimization algorithm for a two-way relay system through numerical simulations. We simulate a flat Rayleigh fading environment where the channel coefficients are randomly generated as zero-mean and unit-variance complex Gaussian random variables. The noise variance $\sigma^2$ is assumed to be unity. For simplicity, the power splitting coefficient $\rho_i, \forall i$, is fixed at $0.5$.

In the first few examples, we demonstrate the effectiveness of the VCG mechanism in self-enforcing truth-telling. Then we provide performance comparison of the proposed joint transmit power and cooperative relay beamforming optimization with some conventional schemes.

For the demonstration of the mechanism design examples, we assign randomly generated values $v_i(g_i)$ instead of calculating $C_{i,{\rm s}}, \forall i,$ which does not affect the relay selection mechanism. It is assumed that although relay $i$ does not know other relays' reported valuation, it knows that every reported value $v_{-i}(g_{-i})$ obeys the probability density function $e^{-x_i}$ where the random variable $x_i \triangleq v_{-i}(g_{-i})$, $x_i \in [0, \infty)$ and $\int_0^{+\infty}e^{-x_i}dx_i = 1$. For simplicity, it is assumed that the price paid per unit of secrecy rate is $\pi_i = 1,~\forall i$.

\begin{figure}
\centering
\includegraphics*[width=\columnwidth]{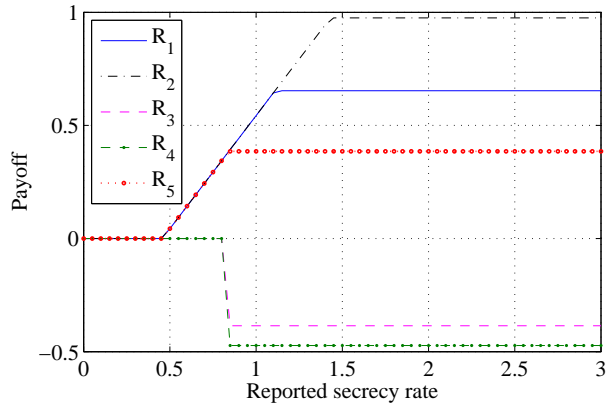}
\caption{Payoff in terms of harvested power (W) versus reported value $x_i$ using VCG mechanism.}\label{fig_vcg1}
\end{figure}

\begin{figure}
\centering
\includegraphics*[width=\columnwidth]{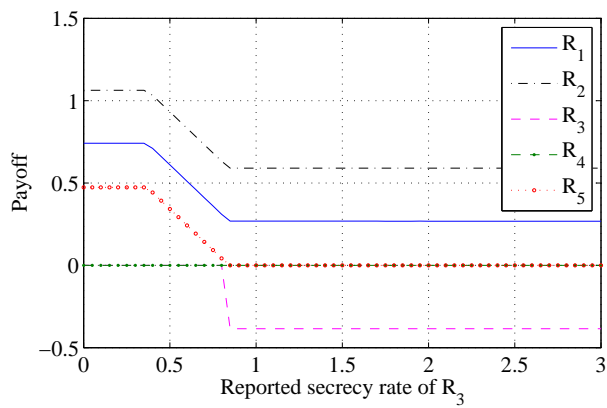}
\caption{Actual payoff of the relays versus reported value of ${\sf R}_3$.}\label{fig_vcgR3}
\end{figure}

In Fig.~\ref{fig_vcg1}, we illustrate how the VCG mechanism works using randomly generated true values of $x_i$'s as $\{1.1101, 1.4321, 0.4567, 0.3690, 0.8421\}$ where the mechanism is to select $K=3$ relays from $N=5$ alternatives. The payoff of each relay node is plotted versus reported $x_i$ values. Note that if all the relays report their true values, then ${\sf R}_1$, ${\sf R}_2$, and ${\sf R}_5$ will be selected and they get their maximum payoff at their true reported values of $1.1101$, $1.4321$, and $0.8421$. It can be observed from Fig.~\ref{fig_vcg1} that both ${\sf R}_1$, ${\sf R}_2$, and ${\sf R}_5$ start receiving positive payoff only after their reported values exceed the highest of the unselected relays' reported values since their selection is not guaranteed otherwise. Also, if either ${\sf R}_3$ or ${\sf R}_4$ reports a value higher than that of ${\sf R}_5$ ($0.8421$), it will be selected instead of ${\sf R}_5$. At that point, the selected relay gets a negative payment which indicates that it needs to use its own source of transmit power for relaying the signal, since it cannot harvest sufficient power due to a poorer actual channel. It is also evident from Fig.~\ref{fig_vcg1} that as long as a relay is not selected, it gets (or pays) nothing.

In the next example, we show the effect of exaggerated reported value by a particular relay (${\sf R}_3$) which is likely to be unselected based on its true channel information assuming that other relays report their true information. Results in Fig.~\ref{fig_vcgR3} illustrate the fact that the exaggerated reported value of ${\sf R}_3$ damages not only its own payoff if selected, but also that of the other relay nodes, which essentially damages the overall system payoff. As discussed in Section~III-A, this is due to the fact that the exaggerated reported value of ${\sf R}_3$ keeps a potential candidate (${\sf R}_5$ in this case) unselected, which results in a higher transfer payment of the selected relays. As soon as the reported value of ${\sf R}_3$ exceeds that of ${\sf R}_5$, it is selected but receives a negative payoff. However, the payoff of ${\sf R}_3$ is always unaffected since there are always some higher reported values than that of ${\sf R}_3$.

\begin{figure}
\centering
\includegraphics*[width=\columnwidth]{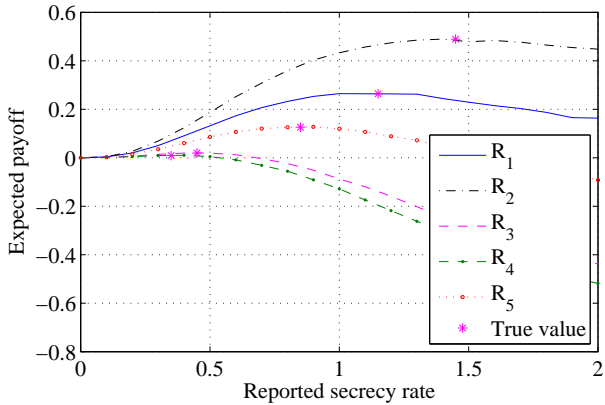}
\caption{Expected payoff of each relay node versus reported value $x_i$.}\label{fig_exVCG}
\end{figure}

Note that the results in Figs.~\ref{fig_vcg1} and \ref{fig_vcgR3} represent the exact payoffs of the relay nodes without taking the probability of being selected into consideration. Hence the payoff of any relay was zero if unselected. However, relays may take the probability of being selected into consideration when deciding which value to report. That will essentially affect their expected payoff as well. In the next example, we intend to show that truth-telling is the best strategy for the relays through their expected payoff where we want to select $K=3$ relays from $N=5$ alternatives. Results in Fig.~\ref{fig_exVCG} show the expected payoff of the relays when their reported values follow negative exponential probability distribution assuming their true affordable secrecy rate of $\{1.1101, 1.4321, 0.4567, 0.3690, 0.8421\}$. We consider a large number ($10^5$) of sample values to calculate the average expected payoff of each relay node at any given reported value. It is now more clearly indicated in Fig.~\ref{fig_exVCG} that truth-telling is the dominant strategy in VCG mechanism. Any agent can expect its maximum payoff only when it reports its true channel information. We can also observe that the larger the true secrecy value of a relay node, the higher the expected payoff. Also, the maximum expected payoff of any relay node is actually less than $u_i\left(\hat{g}_i\right)$ which is because each selected relay node has to pay a mandatory transfer payment as a recovery for the social damage caused by its selection.

The above numerical examples reveal that the VCG mechanism gives the right incentive to the bidders in an auction to disclose their true valuation. Given the mechanism has been implemented perfectly, we now focus on the joint transmit power and cooperative relay beamforming optimization. In order to demonstrate the gain achieved by the proposed SOCP-based joint transmit power and relay beamforming algorithm, we compare the secrecy sum-rate performance of the proposed joint optimization algorithm with that of the relay-only optimization and the conventional randomization-guided semidefinite relaxation (SDR) schemes \cite{multicast, wk_maa_sdr} in the next example. In the relay-only optimization scheme, the two source nodes transmit at fixed power (not optimized). That is, we solve problem \eqref{rateP6} with fixed $p_{{\rm s},1} = p_{{\rm s},2} = \frac{P_{\rm max}}{K+2}$. Note that relay-only optimization is considered for the SDR scheme as well.

In Fig.~\ref{fig_secR}, we compare the secrecy sum rate performance of the proposed algorithm (`Joint opt.' in the figure) with the relay-only optimization (`Relay-only opt.'), and the SDR method of relay beamforming design followed by randomization technique (`SDR approach'). In this example, we select $K=3$ and $4$ relays from a set of $N=8$ alternatives. Note that we initialize the algorithm in Section~IV with $p_{{\rm s},1} = p_{{\rm s},2} = p_{\rm r} = \frac{P_{\rm max}}{K+2}$ and update the transmit powers and relay beamforming vector alternatingly. For updating the transmit powers, we set the tolerable information leakage threshold $r_{\rm e} = 1$ (bps/Hz). Fig.~\ref{fig_secR} shows the performance improvement by the proposed joint optimization algorithm compared to the other two schemes. Since in the randomization approach, some of the constraints may be violated, the performance of the SDR algorithm is severely degraded. For example, at $P_{\rm max} = 10$ dB, the proposed relay-only optimization algorithm achieves more than $1$ bps/Hz higher secrecy sum rate than the randomization approach.

\begin{figure}
\centering
\includegraphics*[width=\columnwidth]{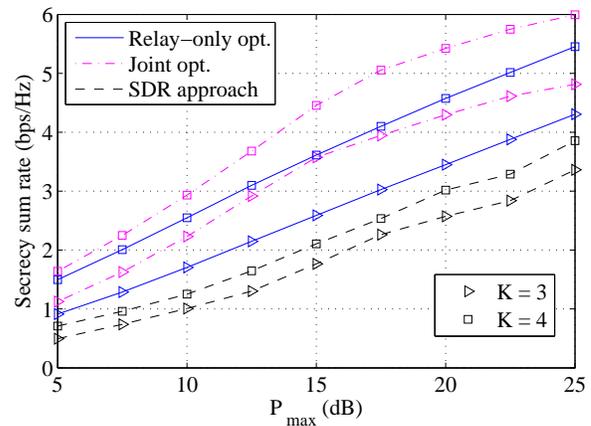}
\caption{Achievable secrecy sum rate versus maximum transmit power with $N =  8$ and $K = 3, 4$.}\label{fig_secR}
\end{figure}

Finally, we show the convergence of the proposed alternating algorithm by evaluating the number of iterations required to converge to an accuracy of $10^{-3}$. We generated four random channel realizations (Channels- $1, 2, 3, 4$) and solved problem \eqref{rateP1}. Fig.~\ref{fig_conv} shows the convergence of the secrecy sum rate maximization problem in different channel realizations with an initial $p_{{\rm s},1} = p_{{\rm s},2} = P_{\rm max}$ for $N = 5$ and $K = 2$. It can be observed that the proposed algorithm achieves a fast convergence in various channel scenarios.
\begin{figure}
\centering
\includegraphics*[width=\columnwidth]{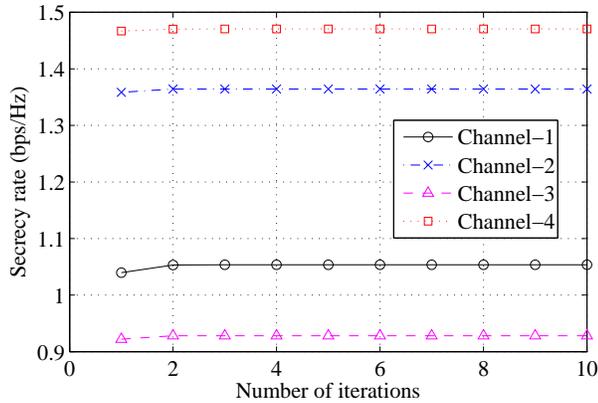}
\caption{Convergence of the proposed two-way relay beamforming algorithm with $N = 5$ and $K = 2$.}\label{fig_conv}
\end{figure}

\section*{\sc VI. Conclusions}\label{sec_con}
In this paper, we considered two-way secret communications via energy harvesting relay nodes. In order to maximize the secrecy rate, the source nodes selected the most suitable relay nodes from the available alternatives. The selected relay nodes, in return, could harvest energy which is guaranteed at least to the minimum payoff level. A self-enforcing truth-telling mechanism design approach was adopted for the relay selection procedure that guarantees that the relays will not exaggerate their true information in order to be selected to gain illegal payoff.
We then proposed a joint cooperative relay beamforming and transmission power optimization algorithm in order to maximize the achievable sum secrecy rate. Designing strategies for dedicated transmit energy beamforming can be an interesting future work.

\bibliographystyle{IEEEtran}\footnotesize{
\bibliography{IEEEabrv,refdb}}

\begin{thebibliography}{10}
\providecommand{\url}[1]{#1}
\csname url@samestyle\endcsname
\providecommand{\newblock}{\relax}
\providecommand{\bibinfo}[2]{#2}
\providecommand{\BIBentrySTDinterwordspacing}{\spaceskip=0pt\relax}
\providecommand{\BIBentryALTinterwordstretchfactor}{4}
\providecommand{\BIBentryALTinterwordspacing}{\spaceskip=\fontdimen2\font plus
\BIBentryALTinterwordstretchfactor\fontdimen3\font minus
  \fontdimen4\font\relax}
\providecommand{\BIBforeignlanguage}[2]{{%
\expandafter\ifx\csname l@#1\endcsname\relax
\typeout{** WARNING: IEEEtran.bst: No hyphenation pattern has been}%
\typeout{** loaded for the language `#1'. Using the pattern for}%
\typeout{** the default language instead.}%
\else
\language=\csname l@#1\endcsname
\fi
#2}}
\providecommand{\BIBdecl}{\relax}
\BIBdecl

\bibitem{jrnl_mur1}
M.~R.~A. Khandaker and Y.~Rong, ``Joint transceiver optimization for multiuser
  {MIMO} relay communication systems,'' \emph{IEEE Trans. Signal Process.},
  vol.~60, pp. 5977--5986, Nov. 2012.

\bibitem{jrnl_su_para}
A.~Toding, M.~R.~A. Khandaker, and Y.~Rong, ``Joint source and relay
  optimization for parallel {MIMO} relay networks,'' \emph{EURASIP J. Adv.
  Signal Process.}, vol. 2012:174, Aug. 2012.

\bibitem{tway_sec_dist}
H.-M. Wang, Q.~Yin, and X.-G. Xia, ``Distributed beamforming for physical-layer
  security of two-way relay networks,'' \emph{IEEE Trans. Signal Process.},
  vol.~60, pp. 3532--3545, July 2012.

\bibitem{tway_sec_jam}
H.-M. Wang, M.~Luo, Q.~Yin, and X.-G. Xia, ``Hybrid cooperative beamforming and
  jamming for physical-layer security of two-way relay networks,'' \emph{IEEE
  Trans. Inf. Forensics and Security}, vol.~8, pp. 2007--2020, Dec. 2013.

\bibitem{tway_rate_reg}
B.~Rankov and A.~Wittneben, ``Achievable rate regions for the two-way relay
  channel,'' in \emph{Proc. IEEE ISIT}, Seattle, USA, 9-14 July 2006,
  1668-1672.

\bibitem{tway_spce_eff}
------, ``Spectral efficient protocols for half-duplex fading relay channels,''
  \emph{IEEE J. Sel. Areas Commun.}, vol.~25, pp. 379--389, Feb. 2007.

\bibitem{jrnl_2way}
M.~R.~A. Khandaker and K.-K. Wong, ``Joint source and relay optimization for
  interference {MIMO} relay networks,'' \emph{EURASIP J. Adv. Signal Process.},
  to appear, 2017.

\bibitem{tway_sec}
E.~Tekin and A.~Yener, ``The general {Gaussian} multiple-access and two-way
  wiretap channels: {Achievable} rates and cooperative jamming,'' \emph{IEEE
  Trans. Inf. Theory}, vol.~54, pp. 2735--2751, June 2008.

\bibitem{secrecy_coop}
L.~Dong, Z.~Han, A.~P. Petropulu, and H.~V. Poor, ``Improving wireless physical
  layer security via cooperating relays,'' \emph{IEEE Trans. Signal Process.},
  vol.~58, pp. 1875--1888, Mar. 2010.

\bibitem{goel_an}
S.~Goel and R.~Negi, ``Guaranteeing secrecy using artificial noise,''
  \emph{IEEE Trans. Wireless Commun.}, vol.~7, pp. 2180--2189, June 2008.

\bibitem{khisti_misome}
A.~Khisti and G.~W. Wornell, ``Secure transmission with multiple antennas {I}:
  The {MISOME} wiretap channel,'' \emph{IEEE Trans. Inf. Theory}, vol.~56, pp.
  3088--3104, July 2010.

\bibitem{qli_spatial}
Q.~Li and W.-K. Ma, ``Spatially selective artificial-noise aided transmit
  optimization for {MISO} multi-eves secrecy rate maximization,'' \emph{IEEE
  Trans. Signal Process.}, vol.~61, pp. 2704--2717, May 2013.

\bibitem{tway_sec_mimo}
J.~Mo, M.~Tao, Y.~Liu, and R.~Wang, ``Secure beamforming for mimo two-way
  communications with an untrusted relay,'' \emph{IEEE Trans. Signal Process.},
  vol.~62, pp. 2185--2199, May 2014.

\bibitem{gth_comm}
J.~Huang, D.~P. Palomar, N.~B. Mandayam, S.~B. Wicker, J.~Walrand, and
  T.~Basar, ``Game theory in communication systems [guest editorial],''
  \emph{IEEE J. Sel. Areas Commun.}, vol.~26, pp. 1042--1046, Sep. 2008.

\bibitem{gth_tway_jamm}
R.~Zhang, L.~Song, Z.~Han, and B.~Jiao, ``Physical layer security for two-way
  untrusted relaying with friendly jammers,'' \emph{IEEE Trans. Vehicular
  Technology}, vol.~61, pp. 3693--3704, Oct. 2012.

\bibitem{gth_int_jamm}
Z.~Han, N.~Marina, M.~Debbah, and A.~Hj{\o}rungnes, ``Physical layer security
  game: {I}nteraction between source, eavesdropper and friendly jammer,''
  \emph{EURASIP J. Wireless Commun. Net.}, vol. 2009, Mar. 2009.

\bibitem{mechD_maga}
B.~Wang, Y.~Wu, Z.~Ji, K.~J.~R. Liu, and T.~C. Clancy, ``Game theoretical
  mechanism design methods,'' \emph{IEEE Signal Process. Mag.}, pp. 74--84,
  Nov. 2008.

\bibitem{mechd_caching}
J.~Dai, F.~Liu, B.~Li, B.~Li, and J.~Liu, ``Collaborative caching in wireless
  video streaming through resource auctions,'' \emph{IEEE J. Sel. Areas
  Commun.}, vol.~30, pp. 458--466, Feb. 2012.

\bibitem{mechD_truth}
J.~Deng, R.~Zhang, L.~Song, Z.~Han, and B.~Jiao, ``Truthful mechanisms for
  secure communication in wireless cooperative system,'' \emph{IEEE Trans.
  Wireless Commun.}, vol.~12, pp. 4236--4245, Sep. 2013.

\bibitem{swipt_1st}
L.~R. Varshney, ``Transporting information and energy simultaneously,'' in
  \emph{Proc. IEEE Int. Symp. Inf. Theory}, July 2008, pp. 1612-1616.

\bibitem{swipt_s2t}
P.~Grover and A.~Sahai, ``Shannon meets {T}esla: Wireless information and power
  transfer,'' in \emph{Proc. IEEE Int. Symp. Inf. Theory}, June 2010, pp.
  2363-2367.

\bibitem{swipt_bc}
R.~Zhang and C.~K. Ho, ``{MIMO} broadcasting for simultaneous wireless
  information and power transfer,'' \emph{IEEE Trans. Wireless Commun.},
  vol.~12, pp. 1989--200, May 2013.

\bibitem{jrnl_swipt}
M.~R.~A. Khandaker and K.-K. Wong, ``{SWIPT} in {MISO} multicasting systems,''
  \emph{IEEE Wireless Commun. Letters}, vol.~3, pp. 277--280, June 2014.

\bibitem{jrnl_secrecy}
------, ``Masked beamforming in the presence of energy-harvesting
  eavesdroppers,'' \emph{IEEE Trans. Inf. Forensics and Security}, vol.~10, pp.
  40--54, Jan. 2015.

\bibitem{jrnl_secrecy_sinr}
------, ``Robust secrecy beamforming with energy-harvesting eavesdroppers,''
  \emph{IEEE Wireless Commun. Letters}, vol.~4, pp. 10--13, Feb. 2015.

\bibitem{swipt_nasir}
A.~A. Nasir, X.~Zhou, S.~Durrani, and R.~A. Kennedy, ``Relaying protocols for
  wireless energy harvesting and information processing,'' \emph{IEEE Trans.
  Wireless Commun.}, vol.~12, pp. 3622--3635, July 2013.

\bibitem{tway_swipt}
Z.~Fang, X.~Yuan, and X.~Wang, ``Distributed energy beamforming for
  simultaneous wireless information and power transfer in the two-way relay
  channel,'' \emph{IEEE Signal Process. Letters}, vol.~22, pp. 656--660, June
  2015.

\bibitem{nobel07}
L.~Hurwicz, \emph{{\rm Optimality and informational efficiency in resource
  allocation processes,} Mathematical Methods in the Social Sciences}.\hskip
  1em plus 0.5em minus 0.4em\relax Stanford University Press, 1960.

\bibitem{vickery}
W.~Vickrey, ``Counterspeculation, auctions, and competitive sealed tenders,''
  \emph{Journal of Finance}, vol. 16(1), pp. 8--37, 1961.

\bibitem{clarke}
E.~Clarke, ``Multi-part pricing of public goods,'' \emph{Public Choice},
  vol.~11, pp. 17--23, 1971.

\bibitem{groves}
T.~Groves, ``Incentives in teams,'' \emph{Econometrica}, vol.~41, pp. 617--631,
  1973.

\bibitem{boyd}
S.~Boyd and L.~Vandenberghe, \emph{Convex Optimization}.\hskip 1em plus 0.5em
  minus 0.4em\relax Cambridge, U.~K.: Cambridge University Press, 2004.

\bibitem{cvx}
M.~Grant and S.~Boyd, ``{CVX}: Matlab software for disciplined convex
  programming (web page and software).'' \texttt{http://cvxr.com/cvx}, Apr.,
  2010.

\bibitem{nemiro_cvx_opt}
A.~Ben-Tal and A.~Nemirovski, \emph{Lectures on modern convex optimization:
  {A}nalysis, algorithms, and engineering applications}.\hskip 1em plus 0.5em
  minus 0.4em\relax MPS SIAM Series on Optimization, Philadelphia, PA, USA:
  SIAM, 2001.

\bibitem{lin_prog_ipm}
N.~Karmarkar, ``A new polynomial-time algorithm for linear programming,''
  \emph{Combinatorica}, vol.~4, pp. 373--396, 1984.

\bibitem{multicast}
N.~D. Sidiropoulos, T.~N. Davidson, and Z.-Q.~T. Luo, ``Transmit beamforming
  for physical-layer multicasting,'' \emph{IEEE Trans. Signal Process.},
  vol.~54, pp. 2239--2251, Jun. 2006.

\bibitem{wk_maa_sdr}
Z.-Q. Luo, W.-K. Ma, M.-C. So, Y.~Ye, and S.~Zhang, ``Semidefinite relaxation
  of quadratic optimization problems,'' \emph{IEEE Signal Process. Mag.},
  vol.~27, pp. 20--34, May 2010.

\end{thebibliography}

\begin{IEEEbiography}
{Muhammad R. A. Khandaker}(S'10-M'13) received the bachelor of science (Honours) degree in computer science and engineering from Jahangirnagar University, Dhaka, Bangladesh, in 2006, the master of science degree in telecommunications engineering from East West University, Dhaka, Bangladesh, in 2007, and the Ph.D. degree in electrical and computer engineering from Curtin University, Australia, in 2013.

Dr. Khandaker has worked in a number of academic positions in Bangladesh. Since July 2013, he has been working as a postdoctoral researcher at the Department of Electronic and Electrical Engineering, University College London (UCL), United Kingdom. He was awarded the Curtin International Postgraduate Research Scholarship (CIPRS) for his Ph.D. study in 2009. He also received the Best Paper Award at the 16th IEEE Asia-Pacific Conference on Communications, Auckland, New Zealand, 2010. He regularly serves in the Technical Program Committees of the IEEE flagship conferences including Globecom, ICC, and VTC. He also served as the Managing Guest Editor of Elsevier Physical Communication, special issue on Self-optimizing Cognitive Radio Technologies. He is currently serving as the Lead Guest Editor of the EURASIP Journal on Wireless Communications and Networking, special issue on Heterogeneous Cloud Radio Access Networks.
\end{IEEEbiography}
\newpage
\begin{IEEEbiography}
{Kai-Kit Wong} (M'01-SM'08-F'16) received the BEng, the MPhil, and the PhD degrees, all in Electrical and Electronic Engineering, from the Hong Kong University of Science and Technology, Hong Kong, in 1996, 1998, and 2001, respectively. He is Professor of Wireless Communications at the Department of Electronic and Electrical Engineering, University College London, United Kingdom. He is Fellow of IEEE and IET. He is Senior Editor of IEEE Communications Letters and also IEEE Wireless Communications Letters.
\end{IEEEbiography}

\begin{IEEEbiography}
{Gan Zheng} (S'05-M'09-SM'12) received the BEng and the MEng from Tianjin University, Tianjin, China, in 2002 and 2004, respectively, both in Electronic and Information Engineering, and the PhD degree in Electrical and Electronic Engineering from The University of Hong Kong in 2008. He is currently a Senior Lecturer in Signal Processing and Networks Research Group in the Wolfson School of Mechanical, Electrical and Manufacturing Engineering, Loughborough University, UK. He had been on various research and visiting positions with University College London, KTH Royal Institute of Technology, and University of Luxembourg, and was with University of Essex as a Lecturer in Communications. His research interests include MIMO precoding, cooperative communications, cognitive radio, physical-layer security, full-duplex radio and energy harvesting. Dr. Zheng is a Senior Member of IEEE. He is the first recipient for the 2013 IEEE Signal Processing Letters Best Paper Award, and he also received 2015 GLOBECOM Best Paper Award. He currently serves as an Associate Editor for IEEE Communications Letters.
\end{IEEEbiography}

\end{document}